\documentclass[12pt, draftclsnofoot, onecolumn]{IEEEtran}
\usepackage{amssymb, amsmath, amsfonts, amsthm, color}
\usepackage{caption, subcaption, cite}
\usepackage{graphicx, cite, mathtools,multicol}
\usepackage{lipsum,graphicx,multicol}
\usepackage{dsfont}
\newtheorem{theorem}{Theorem}

\newtheorem{lemma}{Lemma}

\newtheorem{proposition}{Proposition}

\DeclareMathOperator*{\argmax}{arg\,max}
\newcommand{\diff}{\mathop{}\!d}
\normalsize
\allowdisplaybreaks
\begin{document}

 \title{{\huge Fluid Antenna-aided Full Duplex Communications:\\ A Macroscopic Point-Of-View}}

\author{Christodoulos Skouroumounis, \IEEEmembership{Member, IEEE},
	and~Ioannis Krikidis, \IEEEmembership{Fellow, IEEE}
	\thanks{Christodoulos Skouroumounis and Ioannis Krikidis are with the IRIDA Research Centre for Communication Technologies, Department of Electrical and Computer Engineering, University of Cyprus, Cyprus, e-mail:\{cskour03, krikidis\}@ucy.ac.cy.}}

\maketitle

\begin{abstract}
The synergy of fluid-based reconfigurable antenna (FA) technology and full-duplex (FD) communications can be jointly beneficial, as FD can enhance the spectral efficiency of a point-to-point link, while the new degree of freedom offered by the FA technology can be exploited to handle the overall interference. Hence, in this paper, an analytical framework based on stochastic geometry is developed, aiming to assess both the outage and average sum-rate performance of large-scale FA-aided FD cellular networks. In contrast to existing studies, where perfect channel state information is assumed, the developed framework accurately captures the impact of channel estimation (CE) on the performance of the considered network deployments, as well as the existence of residual loop-interference (LI) at the FD transceivers. Particularly, we focus on a limited coherence interval scenario, where a novel sequential linear minimum-mean-squared-error-based CE method is performed for all FA ports and LI links, followed by data reception from the port with the strongest estimated channel. By using stochastic geometry tools, analytical expressions for the outage and the average sum-rate performance are derived. Our results reveal that FA-aided FD communications experience an improved average sum-rate performance of around 45\% compared to conventional FD communications.
\end{abstract}


\thispagestyle{empty}
\IEEEpeerreviewmaketitle
\newpage
\section{Introduction}\label{Introduction}
\IEEEPARstart{S}{ixth} generation (6G) wireless networks are envisioned to support an enormous number of end-users equipments (UEs) with increasingly demanding spectral efficiency requirements. To meet this demand, next-generation cellular networks will have to foster beneficial synergies between several innovative technologies, such as full-duplex (FD) radio, millimeter wave (mmWave) communications, and reconfigurable antennas \cite{YAD,ALK}. The concept of in-band FD communications can potentially double the spectral efficiency with respect to the half-duplex (HD) counterpart, due to the simultaneous transmission and reception by using non-orthogonal channels \cite{KHAL}. Nevertheless, owing to the non-orthogonal operation of an FD transceiver, a loop-interference (LI) between the output and the input antennas is occurred. As a result of the tremendously negative effect of the LI on transceivers, the FD technology has been previously considered as an unrealistic approach in wireless communications. Fortunately, this long-held pessimistic view has been challenged in the wake of recent advances in antenna design and introduction of analog/digital signal processing techniques \cite{ELB,GOW}.

As a result, FD communications have attracted extensive attention from the research community and the industry, and the potential gains of the FD technology in small-scale networks compared to the HD operation have been clarified by adopting information-theoretic tools to characterize the associated achievable rate regions \cite{SIL}. In the context of large-scale cellular networks, such as multi-cell scenarios, the employment of FD technology leads to simultaneous uplink (UL)/downlink (DL) transmissions on the same time/frequency resource blocks. Hence, FD communications impose both intra- and out-of-cell co-channel interference, jeopardizing the performance of large-scale multi-cell networks. Several research efforts have been carried out to study the effect of loop- and multi-user interference on the FD performance for large-scale wireless networks, and several techniques have been proposed to alleviate the additional structural interference caused by the FD operation. In the context of large-scale FD cellular networks, the authors in \cite{MOH} illustrate that the employment of directional antennas can significantly contribute to the mitigation of the overall interference and the passive suppression of the LI. Moreover, the beneficial effect of multiple antennas in mitigating the additional interference induced by the FD mode is highlighted in \cite{ATZ}, and the feasibility of FD technology in practical scenarios with moderate values of the LI attenuation is demonstrated. The authors in \cite{LEE2} investigate the concept of hybrid HD/FD cellular networks, indicating that the enhancement of the network throughput is achieved by the operation of different tiers at different duplex modes. In \cite{SKO1}, a location-based FD/HD scheduling method is developed for mmWave communications, illustrating that HD mode is beneficial for the cell-edge UEs to achieve better network performance, as opposed to the cell-center UEs where the FD mode is more efficient. Finally, the beneficial combination of FD radio with heterogeneous mmWave systems is investigated in \cite{SKO2}, where the prominent properties of mmWave communications are leveraged towards combating the severe multi-user interference caused by the FD technology.

Fluid-based reconfigurable antennas (FAs) have been envisioned as a promising technology in future communication devices due to their attractive features such as conformability, flexibility, and reconfigurability \cite{AKY}. In particular, the concept of FAs refers to liquid radiating elements (e.g., mercury, eutectic gallium indium (EGaIn), galinstan, etc.), which are contained in a dielectric holder. With the assistance of a dedicated microelectromechanical system (MEMS), the liquid radiating element can flow in different locations (i.e., a set of predefined ports) within its topological boundaries \cite{HUA}. Hence, FAs are capable to reversibly re-configure their physical configuration (i.e., size, shape and feeding) as well as their electrical properties (e.g., resonant frequency, bandwidth), offering a new degree of freedom in the design of wireless communication systems \cite{TAR}. Thereby, the concept of FA technology has spurred the interest of the research community and the industry, where FA-enabled communications have already been evaluated under various communication scenarios and the wide range of functionalities attained by the employment of reconfigurable FAs has been investigated \cite{BOR,SIN,SHE}. The outage performance and the ergodic capacity of FA-enabled point-to-point communication systems under spatially correlated channels are investigated in \cite{WON1} and \cite{WON2}, respectively; it is shown that a single FA with half-wavelength or less separation between ports can achieve capacity and outage performance similar to the conventional multi-antenna maximum ratio combining system, if the number of ports is sufficiently large. These works are further extended into multi-user communications in \cite{WON3}, where a mathematical framework that takes into account multiple pairs of transmitters and FA-based receivers is introduced. Based on the developed mathematical framework, the authors investigate the outage performance, the achievable rate, and the multiplexing gain of the considered topology, revealing that the overall network performance improves as the number of ports increases at each receiver. Finally, in \cite{KHA}, the outage performance of FA-enabled communications is investigated under a more accurate spatial correlation model between the FAs' ports, while maintaining the analytical tractability of the developed framework. All the above-mentioned studies have also shown that leveraging the spatial diversity offered by the FA concept, a new degree of freedom to handle multi-user interference is provided e.g., selecting the FA port that either maximizes the signal-to-interference-plus-noise ratio (SINR) or suffers from strong multi-user interference that enables interference cancellation techniques. Thus, the synergy of FD radio and FA-enabled networks is of paramount importance to combat multi-user and LI interference (induced by the FD technology) by exploiting the additional degree of freedom offered by the FA technology. However, even though FD radio and FA technology have been well-investigated in several works, their co-design has been disregarded. 

An indispensable prerequisite towards unleashing the full potential of FA-aided FD communications is the existence of accurate channel estimation (CE). More specifically, the acquisition of perfect channel state information (CSI) for the network's direct links i.e., the links between a transmitter and a receiver, enables the optimal selection of the location of the FA's liquid radiating element with the objective to mitigate the multi-user interference. Furthermore, the successful cancellation of the LI strongly depends on the quality of the estimated channel for the LI link i.e., the link between the transmit and the receive antennas at the FD transceivers. In practice, such CSI needs to be acquired in each channel coherence interval at the cost of a channel training overhead that escalates with the number of FA ports \cite{YAZ}. In the context of a limited coherence interval scenario, this training period inevitably leads to a smaller data transmission duration and thus, a reduced overall network performance. Hence, a non-trivial trade-off between the channel training duration and the overall network performance is triggered owing to the CE processes performed in the large-scale FA-aided FD communications. In the context of deterministic network deployments, several works have already investigated the impact of CE on the network performance under different communication scenarios, including basic multiple-input multiple-output (MIMO) settings \cite{LOU,KAY} and cellular networks \cite{JOS,YIN}. However, the aforementioned works do not capture the randomness of actual cellular network deployments. In such a context, system-level performance evaluation will be critical to obtain relevant insights into the design trade-offs that govern such complex systems. Over the past decade, stochastic geometry (SG) has emerged as a powerful mathematical tool, which captures the random nature of large-scale networks \cite{SGbook}. {Due to this fact, stochastic geometry is considered as a key tool to model, analyze and design current cellular networks, permitting the analytical characterization of numerous performance metrics.} Despite the plethora of works evaluating the impact of CE on the performance of a small-scale networks, such as single-cell scenarios, a profound understanding of the effects on large-scale network deployments is however still elusive. In \cite{WU2}, the impact of CE on the coverage performance of random networks is evaluated, capturing the dependence of the optimal training-pilot length on the ratio between the receiver and transmitter deployment densities. In the context of FA systems, in \cite{SKO}, the authors assess the effect of FA technology on large-scale cellular networks, and study the trade-off imposed by the CE on the outage performance; it is unveiled that an optimal number of FAs' ports maximizes the network performance with respect to the linear minimum-mean-squared-error (LMMSE)-based CE process.

Motivated by the above, in this work, we investigate the achieved performance of large-scale FA-aided FD cellular networks under a limited coherence interval scenario, and study the trade-off imposed by the CE on the network performance. We adopt a system-level point-of-view by considering the spatial randomness of both the BSs and the UEs, providing a rigorous mathematical framework to analyze the performance of the considered network topology, in terms of outage probability and average sum-rate. Specifically, the main contributions of this paper are summarized as follows:
\begin{itemize}
	\item {We propose an analytical framework based on SG, which comprises the synergy of FA technology and FD radio, shedding light on the modeling, design, and analysis of large-scale FA-aided FD communication. In particular, the developed mathematical framework characterizes the performance limits of FA-aided FD communication networks in single-antenna systems, and serves as a guideline for more complex multi-antenna network topologies.} Moreover, the developed framework takes into account the capability of both BSs and UEs to operate in FD mode, while all UEs are equipped with a single FA and employ distance-proportional fractional power control. Based on the developed framework, the performance of the considered deployment is evaluated in terms of outage probability and average sum-rate performance, under a limited coherence interval scenario.
	\item Building on the developed mathematical framework, a novel and low-complexity LMMSE-based CE technique for all direct and LI links is proposed. Initially, the CE between the UEs' ports and their serving BSs is performed via pilot-training symbols in a sequential manner, by taking into account the spatial correlation between FAs' ports. Thereafter, the LMMSE-based CE for the LI links at both the BSs and the UEs is acquired. Finally, by aiming to achieve the best end-to-end performance, the FD data transmission and reception from the port with the strongest estimated channel is performed.
	\item By using SG tools, analytical expressions for the outage and the sum-rate performance achieved by the considered network deployments are derived. These analytical expressions provide a quick and convenient method of evaluating the system's performance and obtaining insights into how key parameters affect the performance. In addition, our results demonstrate that an optimal number of FAs' ports maximizes the network performance with respect to the LMMSE-based CE technique. Finally, our results highlight the beneficial synergy of FA technology and FD radio, providing an increase of the achieved average sum-rate performance by around 45\% compared with the conventional static FD communications.
\end{itemize}

\textit{Organization:} In Section \ref{SystemModel}, system model and modeling assumptions are introduced. Section \ref{Preliminaries} presents the statistical properties of the observed interference in the considered network deployments, as well as, the CE and the data transmission operations. Analytical expressions for the performance achieved by the FA-aided FD cellular networks are evaluated in Section \ref{PerformanceAnalysis}. Simulation results are presented in Section \ref{Numerical}, followed by our conclusions in Section \ref{Conclusion}. 
\begin{table*}[t]\centering
	\caption{Summary of Notations}\label{Table1}
	\scalebox{0.85}{
		\begin{tabular}{| l | l || l | l |}\hline
			\textbf{Notation} 		& \textbf{Description} 									& \textbf{Notation} 						& \textbf{Description}\\\hline
			$\Phi,\lambda_b$ 		& PPP of BSs of density $\lambda_b$ 					&  $B_c, T_c$								& Coherence bandwidth and time  \\\hline
			$\Psi$, $\lambda_u$ 	& Point process of UEs of density $\lambda_u$  			&  $L_e, L_t$								& Channel estimation and data transmission period  \\\hline
			$N,\mathcal{N}$ 		& Number and set of FA's ports 							&  $L_d,L_{\rm LI}$						& Channel estimation period for direct and LI links   \\\hline
			$\kappa$, $\lambda$ 	& Scaling constant and communication wavelength 		&  $l_s$									& Switching channel uses  \\\hline
			$d_i$ 					& Displacement of the $i$-th port  						&  $\Lambda$									& Number of pilot-training symbols for each port\\\hline
			$u,\delta$ 			& Average velocity and delay of fluid metal  	&  $f_\mathcal{I}(\cdot),\varpi,\varrho$	& Gamma distribution with parameters $\varpi$ and $\varrho$ \\\hline
			$r_i(\rho)$ 			& Distance between the $i$-th port and the serving BS 	&  $P_u(R), P_{\rm m}$	& Transmit power of UEs with constraint $P_{\rm m}$ \\\hline
			$P$ 					& Transmission power  									&  $\epsilon, \omega$							& Power control parameters \\\hline
			$\ell(r)$ 				& Large scale path loss 								&  SINR$_i$									& SINR observed at the $i$-th port \\\hline
			$\mu_i$ 				& Autocorrelation parameter 							&  $\sigma_{e_i}^2$							& Variance of the channel estimation error for direct link \\\hline
			$L_c$					& Channel uses in each coherence interval 				&  $\sigma_{e_{\rm LI}^u}^2,\sigma_{e_{\rm LI}^b}^2$						& Variance of the channel estimation error for UE's and BS's LI link \\\hline
	\end{tabular}}\vspace{-15pt}
\end{table*}
\section{System Model}\label{SystemModel}
In this section, the details of the considered system model are provided. A list of the main mathematical notations is presented in Table \ref{Table1}.
\subsection{Network Topology}
Consider a homogeneous cellular network, where the BSs are uniformly distributed in $\mathbb{R}^2$ according to independent homogeneous Poisson point process (PPP) $\Phi$, with density $\lambda_b$. We also consider a set of UEs, whose locations follow an arbitrary independent point process with spatial density $\lambda_u\gg  \lambda_b$. We assume that all BSs are equipped with a single omnidirectional antenna, while all UEs are equipped with a FA (detailed description in Section \ref{FAmodel}). Furthermore, we consider the scenario where all BSs and UEs operate in FD mode. {Since multiple UEs can exist in the coverage area of a BS, a scheduling mechanism is employed, which schedules all UEs for their communication with the assigned BS at different time-frequency resources. That means no intra-cell interference exists since intra-cell users are served with orthogonal time-frequency resources.} We consider the nearest-BS association rule i.e., the typical UE at the origin communicates in both DL and UL transmission with its closest BS located at $x_{\rm 0}\in\mathbb{R}^2$, referred as \textit{tagged BS}, and its link with the typical UE is denoted as \textit{typical link} (see Fig. \ref{Voronoi}). Hence, the random variable representing the distance between the typical UE and the tagged BS i.e., $\rho=\|x_0\|$, follows a probability density function (pdf), that is given by \cite{SGbook}
\begin{equation}\label{ContactPDF}
	f_R(\rho) = \frac{{\rm d}\mathbb{P}[R>\rho]}{{\rm d}\rho}=2\pi\lambda_b\rho\exp\left(-\pi\lambda_b\rho^2\right),
\end{equation}
where $\mathbb{P}[R>\rho]$ is the complementary cumulative distribution function of $R$, that is given by $\mathbb{P}[R>\rho]=\exp\left(-\pi\lambda_b\rho^2\right)$.

\begin{figure}
	\centering\includegraphics[width=0.45\linewidth]{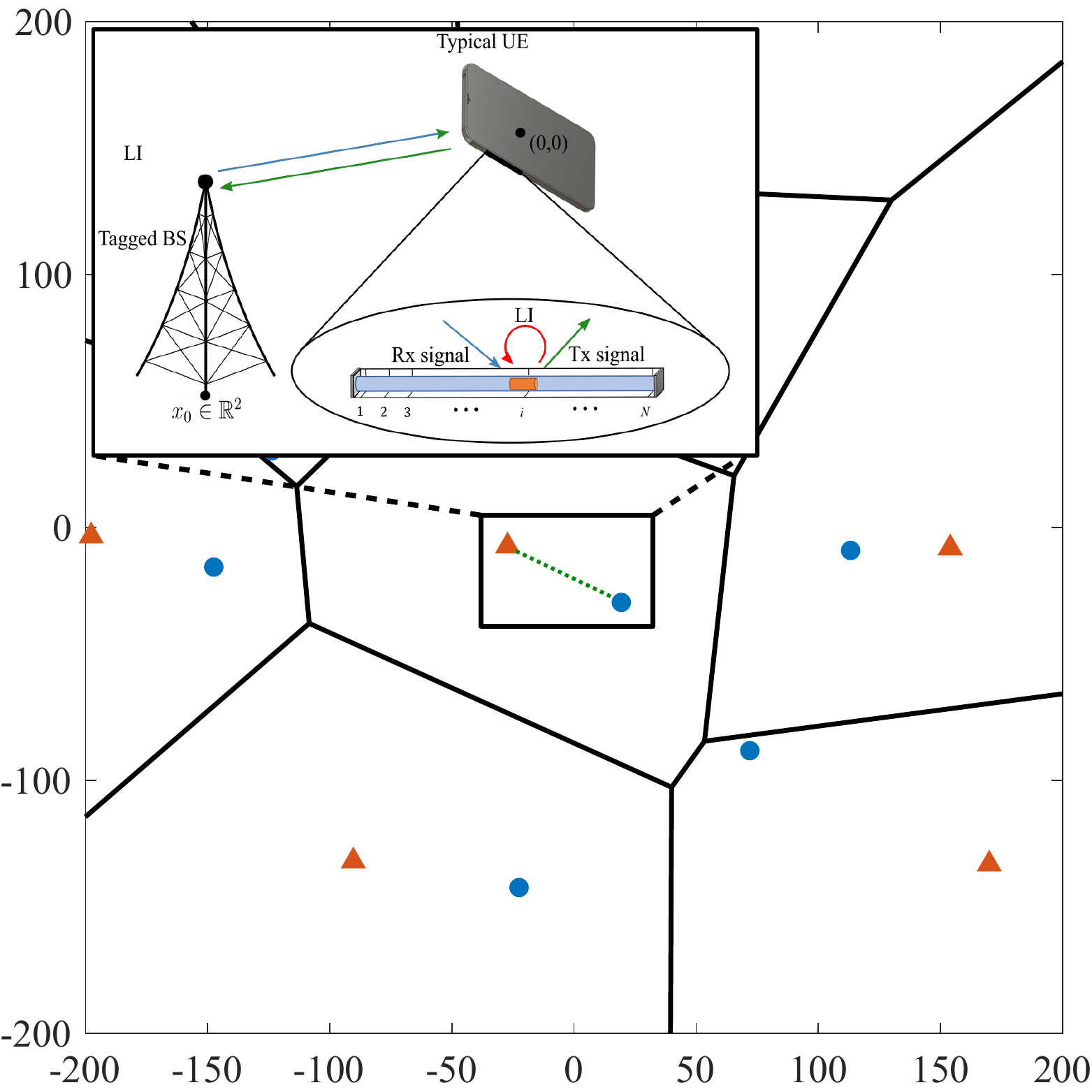} 
	\caption{The Voronoi tessellation of a FA-aided FD cellular network, where the BSs and the UEs are represented by triangles and circles, respectively. Dotted line represents the typical link.}\label{Voronoi}\vspace{-20pt}
\end{figure}
\subsection{Fluid antenna model}\label{FAmodel}
The adopted architecture of a fluid-based reconfigurable antenna is depicted in Fig. \ref{FA}. In particular, a drop of liquid metal (e.g., EGaIn) is positioned in a tube-like linear micro-channel or capillary filled with an electrolyte, within which the fluid is able to move freely. More specifically, the location of the antenna (i.e., fluid metal) can be switched to one of the $N$ predetermined locations (also known as ``ports''), that are evenly distributed along the linear dimension of a FA, $\kappa\lambda$, where $\lambda$ is the wavelength of communication and $\kappa$ is a scaling constant. An abstraction of the FA concept is considered, where an antenna at a given location is treated as an ideal point antenna \cite{WON1}. By aiming to ease the mathematical analysis, the first port of a FA is treated as an auxiliary reference port \cite{WON1}. As depicted in Fig, \ref{FA}, the displacement between the first port and the $i$-th port can be measured as
\begin{equation*}
d_i = \left(\frac{i-1}{N-1}\right)\kappa\lambda,\ \forall\ i\in\mathcal{N},
\end{equation*}
where $\mathcal{N}=\{1,2,\dots,N\}$. Hence, the distance between the $i$-th port of the typical UE and the tagged BS, is given by
\begin{equation}\label{Distance}
r_{i}(\rho) = \sqrt{\rho^2+\frac{\kappa^2\lambda^2}{4}\left(\frac{N-2i+1}{N-1}\right)^2},\ \forall\ i\in\mathcal{N},
\end{equation}
where $\rho$ is the distance of the typical link. {It is important to mention here that, by considering $N=1$, the adopted FA model models the conventional (i.e., static) omnidirectional antenna, enabling the evaluation of the performance achieved by the conventional FD communications.}

Regarding the flow motion of a fluid metal within an electrolyte-filled capillary, we assume that each UE is equipped with a MEMS. In particular, the flow motion of the fluid metal is induced by the application of a voltage gradient along the FA, as shown in Fig. \ref{FA}, as a consequence of the electrocapillary effect \cite{LEE}. Hence, in accordance to the Hagen–Poiseuille equation \cite{LEE}, the fluid metal achieves an average velocity that can be evaluated as
\begin{equation*}
u=\frac{q}{6\mu}\frac{D}{L}\Delta\phi,
\end{equation*}
where $q$ denotes the initial charge in the electrical double layer for the liquid radiating element (e.g. EGaIn), $\Delta\phi$ represents the voltage difference between the two ends of the fluid metal which is much smaller than the externally applied voltage $U$ i.e., $\Delta\phi\ll U$, $\mu$ is the viscosity of the liquid radiating element at $20^oC$, $D$ and $L$ represent the thickness and the length of the fluid metal, respectively. Due to the finite velocity of fluid metals, a non-zero time period is required for the movement of the fluid metal from a port to another, referred as ``delay'', opposed to the impractical assumption of instantaneous port switching that is widely-adopted in the existing works by considering high-velocity fluid metals. In particular, the time required (delay) by the fluid metal to move from the $i$-th port to the $(i+1)$-th port, is given by
\begin{equation*}
\delta =\frac{\kappa\lambda}{u}\left(\frac{1}{N-1}\right),
\end{equation*}
where $i\in\mathcal{N}$. \footnote{For example, a fluid metal velocity of $116$ mm/s is obtained by assuming $\Delta\phi=0.1$ V and $L/D=5$ \cite{GOU}, and hence, a delay of $54\ \mu$s is experienced between neighbouring ports of a FA with $N=20$, $\kappa = 0.2$, and $\lambda=0.06$ cm \cite{WON1}.}

\begin{figure}[t]
	\centering\includegraphics[width=0.3\linewidth]{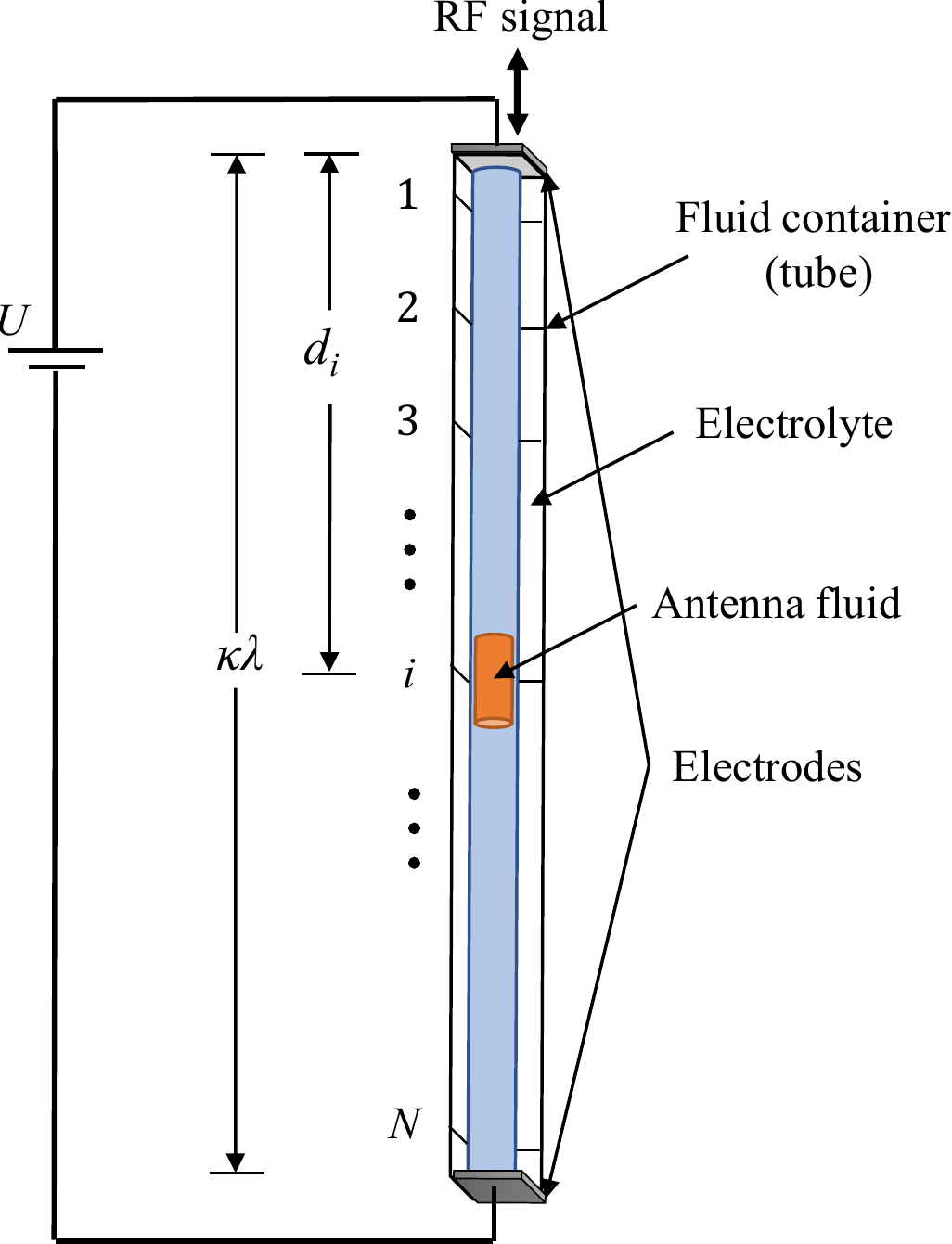}
	\caption{A potential FA architecture.}\label{FA}\vspace{-20pt}
\end{figure}

\subsection{Channel model}
We assume that all wireless signals experience both large-scale path-loss effects and small-scale fading. More specifically, the large-scale attenuation incurred at the transmitted signals follows an unbounded singular path-loss model i.e., $\ell(r)=r^{-a}$, which assumes that the received power decays with the distance $r$ between the transmitter located at $X$ and the receiver located at $Y$ i.e., $r = \|X-Y\|$, where $a > 2$ denotes the path-loss exponent. 
For the small-scale fading, a block fading channel model is considered. Particularly,  we assume that the channel remains constant during a coherence time $T_c$, also known as \textit{channel coherence interval}, and evolves independently from block to block. Specifically, we assume that the small-scale fading between two nodes is modeled as a circularly symmetric complex Gaussian distribution with zero mean and variance of $\sigma^2$, and thus, the channel's amplitude between the $i$-th port of the typical UE and its serving BS, $\left| g_{{\rm 0}i}\right| $, is Rayleigh distributed\footnote{While such an assumption clearly does not reflect a real network environment, it still enables us to obtain some closed-form expressions which may be used as estimates for more realistic situations.}. Owing to the capability of ports being arbitrarily close to each other, the observed channels are considered to be correlated \cite{LIAN}. In particular, the channels observed by the $N$ ports of the typical UE, can be evaluated as
\begin{equation*}
g_{{\rm 0}i}=\begin{cases}
\sigma \alpha_1+j\sigma \beta_1 &\text{if $i=1$},\\
\!\sigma\!\left(\!\sqrt{\!1\!-\!\mu_i^2}\alpha_i\!+\!\mu_i \alpha_0\!\right)\!\!+\!j\sigma\!\left(\!\sqrt{1\!-\!\mu_i^2}\beta_i\!+\!\mu_i \beta_0\!\right) &\text{otherwise},\\
\end{cases}
\end{equation*}
where $i\in\mathcal{N}$, $\alpha_1,\dots,\alpha_N,\beta_1,\dots,\beta_N$ are all independent Gaussian random variables with zero mean and variance of $\frac{1}{2}$; $\mu_i$ is the autocorrelation parameter that can be chosen appropriately to determine the channel correlation between the $i$-th and the reference (i.e., first) port of a FA. Specifically, we assume that the autocorrelation parameter is given by \cite{WON1}
	\begin{equation*}
	\mu_i = \begin{cases}
	0&\text{if $i=1$},\\
	J_0\left(\frac{2\pi(i-1)}{N-1}\kappa\right) &\text{otherwise},\\
	\end{cases}
	\end{equation*}
	where $J_0(\cdot)$ is the zero-th Bessel function of the first kind \cite{WON2}. 

Regarding the LI, we assume that all FD receivers (i.e., UE or BS) employ imperfect cancellation mechanisms \cite{ALK}. Such cancellation mechanisms rely heavily on the accuracy of the LI links' CE \cite{ALK}. As such, we consider the residual LI channel coefficient at the BSs and the UEs to follow a Nagakami-$\mu$ distribution with parameters $\left(\mu,\sigma^2_{e_{\rm LI}^b}\right)$ and $\left(\mu,\sigma^2_{e_{\rm LI}^u}\right)$, respectively. Note that, $\sigma^2_{e_{\rm LI}^b}$ and $\sigma^2_{e_{\rm LI}^u}$ depict the quality of the CE for the LI links at the BSs and the UEs, respectively. Therefore, the power gain of the residual LI channel at the BSs and the UEs follows a Gamma distribution with mean $\mu$ and variance $\sigma^2_{e_{\rm LI}^b}/\mu$ and $\sigma^2_{e_{\rm LI}^u}/\mu$, respectively i.e., $|\hat{h}^b_{\rm LI}|^2\sim\Gamma\left(\mu,\frac{\sigma^2_{e_{\rm LI}^b}}{\mu}\right)$ and $|\hat{h}^u_{\rm LI}|^2\sim\Gamma\left(\mu,\frac{\sigma^2_{e_{\rm LI}^u}}{\mu}\right)$. 
\subsection{Power allocation}
Due to the scarce power resources of battery-powered devices, the UL power control is of paramount importance in next-generation cellular networks \cite{SES}. Hence, the employment of a power control scheme only for the UL transmission is considered, while for the DL transmission we assume a fixed power transmission allocation scheme. In particular, all UEs utilize a distance-proportional fractional power control in order to compensate the large-scale path-loss and maintain the average received signal power at their corresponding serving BSs equal to $\omega$ \cite{SAK}. To achieve this, a UE located at a distance $R$ from its serving BS, and thus, with path-loss equal to $R^{-a}$, adapts its transmitted signal power to $\omega R^{a \epsilon}$, where $0\leq\epsilon\leq 1$ is the power control fraction. According to the general UL power control mechanism used in the Long-Term-Evolution (LTE) standard \cite{SES}, the transmission power allocated to a cellular UE can be expressed as $P_u(R)=\min\{\omega R^{a \epsilon},\ P_{\rm m}\}$, where the UEs which are unable to fully invert the path-loss, transmit with maximum power $P_{\rm m}$. Note that, for the case where $\epsilon\!=\!1$, the path-loss is completely compensated if $P_{\rm m}$ is sufficiently large, and if $\epsilon=0$, no path-loss inversion is performed and all the UEs transmit with the same power. For the DL transmissions, we consider a fixed power transmission allocation scheme i.e., BSs transmit with power $P$.

\section{Channel Estimation and Data Communication Under a Limited Coherence Interval}\label{Preliminaries}
In this section, we initially evaluate the statistic properties of the aggregate interference at the typical FD receiver. Moreover, we elaborate in detail the CE and data transmission periods in the context of the adopted LMMSE-based CE technique, as well as the considered imperfect SI cancellation scheme. We finally provide analytical expressions which will be useful for computing the outage and average sum-rate performance in Section \ref{PerformanceAnalysis}.

\subsection{Interference Characterization}\label{MUInterf}
In this section, we examine the received interference at the typical receiver (i.e., UE or BS), where analytical expression for the mean aggregate interference is derived. Note that, due to the existence of both DL and UL transmissions, the aggregate interference at the typical receiver is induced by both the BSs and the UEs. Initially, the interference caused by the BSs and observed at the $i$-th port of the typical UE and the tagged BS, can be expressed as follows
\begin{equation}\label{BS-UE}
	\left(\mathcal{I}_i^{{\rm DL}}\right)_{\rm BS}=P\sum\nolimits_{\substack{t\in\mathbb{N}^+ \\ x_t\in\Phi\backslash x_0}} \ell(r_{i}(\|x_t\|))|g_{ti}|^2\text{ and } \left(\mathcal{I}_i^{\rm UL}\right)_{\rm BS}=P\sum\nolimits_{\substack{t\in\mathbb{N}^+ \\ x_t\in\Phi}} \ell(r_{i}(\|x_t\|))|g_{ti}|^2,
\end{equation}
respectively. In addition, the interference induced by the active UEs and observed at the $i$-th port of the typical UE and the tagged BS, is given by
\begin{equation}\label{UE-UE}
	\left(\mathcal{I}_i^{{\rm DL}}\right)_{\rm UE}=\sum\nolimits_{\substack{t\in\mathbb{N}^+ \\ y_t\in\Psi}} P_u(\|\overline{y_t}\|)\ell(r_{i}(\|y_t\|))|h_{ti}|^2,
\end{equation}
and
\begin{equation}\label{UE-BS}
	\left(\mathcal{I}_i^{\rm UL}\right)_{\rm UE}=\sum\nolimits_{\substack{t\in\mathbb{N}^+ \\ y_t\in\Psi\backslash x_0}} P_u(\|\overline{y_t}\|)\ell(r_{i}(\|y_t\|))|h_{ti}|^2,
\end{equation}
respectively, where $P_u(\|\overline{y}\|)=\min\{\omega \|\overline{y}\|^{a \epsilon},P_{\rm m}\}$ represents the transmit power of the interfering UE to its serving UL BS located at $\overline{y}$; $h_{ti}$ is the channel fading between the typical receiver and the $i$-th port of the interfering UE at $y_t$, and $\Psi$ is the point process that represents the active UEs. Therefore, the aggregate interference observed for the DL and the UL transmissions are given by
\begin{equation}\label{OverallDL}
	\mathcal{I}_i^{{\rm DL}} = \left(\mathcal{I}_i^{{\rm DL}}\right)_{\rm BS}+\left(\mathcal{I}_i^{{\rm DL}}\right)_{\rm UE}
\end{equation}
and 
\begin{equation}\label{OverallUL}
	\mathcal{I}_i^{\rm UL} = \left(\mathcal{I}_i^{\rm UL}\right)_{\rm BS}+\left(\mathcal{I}_i^{\rm UL}\right)_{\rm UE},
\end{equation}
respectively.

With the aim of defining the received interference for the DL/UL transmission, the characterization of interference-free areas is essential. In particular, the interference-free area of each typical receiver (i.e, BS or UE) is quantified by defining the spatial density of the interfering nodes $f(r)$.
\begin{itemize}
	\item BSs-to-UEs interference: According to the adopted association criterion, all BSs that cause interference to the typical UE exhibit greater path-loss compared to the path-loss of the tagged BS. Thus, the locations of the interfering BSs must satisfy the condition $\|x\|>\|x_0\|$, where $x\in\Phi$ and $x_0\in\Phi$. Hence, $f_1(r) = \lambda_b\mathds{1}\left\lbrace r > \|x_0\|\right\rbrace$.
	\item UEs-to-BSs interference: There is no exact boundary for the interference-free area around the tagged BS, since the interfering UEs can be arbitrarily close. However, based on the adopted criterion, we can conclude that for a UE which is located at $y$ and is served by a BS at $x\in\Phi$, $\|x_0-y\|\!<\!\|x-y\|$. Hence, $f_2(r)\!=\!1\lambda_b\mathds{1}\left\lbrace r \!>\! R\right\rbrace$, where $R=\|x_0-y\|$.
	\item BSs-to-BSs interference: Owing to the PPP assumption, the BSs can be very close to each other. In real-world deployments, however, this is not true due to physical constraints, different antenna heights, etc. Based on this, we adopt the approximation proposed in \cite{SAK}, where the locations of the interfering BSs follow a non-homogeneous point process with $f_3(r)=\lambda_b\left(1-\exp\left(-\pi \lambda_b r^2\right)\right)\mathds{1}\left\lbrace r > b_b\right\rbrace$, where $b_b$ is a constraint on the length of an BS-to-UE interference link.
	\item UEs-to-UEs interference: We adopt a similar approximation to model the interfering UEs. In particular, the point process $\Psi$ is modelled as an inhomogeneous PPP with density $f_4(r) = \lambda_b\left(1-\exp\left(-\pi \lambda_b r^2\right)\right)\mathds{1}\left\lbrace r > b_u\right\rbrace$, where $b_u$ is a constraint on the length of an UE-to-UE interference link.
\end{itemize}  
To characterize the network interference, in the following Lemmas, we compute the mean of the random variables $I^{\rm DL}_i$ and $I^{\rm UL}$.

\begin{lemma}\label{MeanInterferenceUL}
	The mean of $\left(I^{\rm UL}\right)_{\rm BS}$, conditioned on the distance from the serving transmitter i.e., $\rho$, is given by
	\begin{equation}\label{MeanULBS}
		\mathbb{E}\!\left(\left(\mathcal{I}_i^{\rm UL}\right)_{\rm BS}|\rho\right)\!=\!\frac{\pi\sigma^2b_b^{2-a}}{a-2}\!P\lambda_b\!\left((a\!-\!2)E_{\frac{a}{2}}[\pi\lambda_b b_b^2]\!+\!2\right),
	\end{equation}
	and the mean of $\left(I^{\rm UL}\right)_{\rm UE}$, conditioned on $\rho$, is given by
	\begin{equation}\label{MeanULUE}
		\mathbb{E}\!\left(\left(\mathcal{I}_i^{\rm UL}\right)_{\rm UE}|\rho\right)\!=\!\frac{2\pi^\frac{a}{2}}{a\!-\!2}\!P\lambda_b^\frac{a}{2}\!\Bigg(\!P_{\rm m}\Gamma\!\left[2\!-\!\frac{a}{2},\pi\lambda_b\!\left(\!\frac{P_{\rm m}}{\omega}\!\right)^\frac{2}{a\epsilon}\right]-\pi^{-\frac{a\epsilon}{2}}\omega\lambda_b^{1-\epsilon}\Delta\Gamma\left(2-\frac{a}{2}(\epsilon-1)\right)\Bigg),
	\end{equation}
	where $\Delta\Gamma(x) = \Gamma\left[x,\pi\lambda_b\left(\frac{P_{\rm m}}{\omega}\right)^\frac{2}{a\epsilon}\right]-\Gamma[x]$ and $E_\alpha[\beta]$ is the exponential integral function \cite{Marcumbook}. Thus, the mean of random variable $\mathcal{I}_i^{\rm UL}$, condition on $\rho$, is given by
	\begin{equation*}
		\mathbb{E}\left(\mathcal{I}_i^{\rm UL}|\rho\right) = \mathbb{E}\left(\left(\mathcal{I}_i^{\rm UL}\right)_{\rm BS}|\rho\right)+\mathbb{E}\left(\left(\mathcal{I}_i^{\rm UL}\right)_{\rm UE}|\rho\right).
	\end{equation*}
\end{lemma}
\begin{proof}
	See Appendix \ref{Appendix4}.
\end{proof}
\begin{lemma}\label{MeanInterferenceDL}
	The mean $\left(I^{\rm DL}_i\right)_{\rm BS}$, conditioned on $\rho$, is given by
	\begin{equation}\label{MeanDLBS}
		\mathbb{E}\left(\left(\mathcal{I}_i^{\rm DL}\right)_{\rm BS}|\rho\right)=\frac{2\pi\sigma^2}{a-2}P\lambda_b r_{i}^{2-a}(\rho),
	\end{equation}
	and the mean of $\left(I^{\rm DL}_i\right)_{\rm UE}$, conditioned on $\rho$, is given by
	\begin{align}\label{MeanDLUE}
		\mathbb{E}\left(\left(\mathcal{I}_i^{\rm DL}\right)_{\rm UE}|\rho\right)&\!=\!\frac{\pi \sigma^2b_u^{2-a}}{a-2}\!P\lambda_b\!\left((a\!-\!2)E_{\frac{a}{2}}[\pi\lambda_b b_u^2]\!-\!2\right)\nonumber\\
		&\qquad\qquad\times\!\left(\!\frac{\omega\Delta\Gamma\left(1+\frac{a\epsilon}{2}\right)}{(\pi\lambda_b)^{\frac{a\epsilon}{2}}}\!-\!P_{\rm m}\exp\!\left(\!-\pi\lambda_b \!\left(\!\frac{P_{\rm m}}{\omega}\!\right)^{\frac{2}{a\epsilon}}\!\right)\right),
	\end{align}
	where $\Delta\Gamma(x) = \Gamma\left[x,\pi\lambda_b\left(\frac{P_{\rm m}}{\omega}\right)^\frac{2}{a\epsilon}\right]-\Gamma[x]$ and $E_\alpha[\beta]$ is the exponential integral function. Thus, the mean of random variable $\mathcal{I}_i^{\rm DL}$, condition on $\rho$, is given by
	\begin{equation*}
		\mathbb{E}\left(\mathcal{I}_i^{\rm DL}|\rho\right) = \mathbb{E}\left(\left(\mathcal{I}_i^{\rm DL}\right)_{\rm BS}|\rho\right)+\mathbb{E}\left(\left(\mathcal{I}_i^{\rm DL}\right)_{\rm UE}|\rho\right).
	\end{equation*}
\end{lemma}
\begin{proof}
	The proof follows similar arguments with the proof for the mean of $\mathbb{E}\!\left(\left(\mathcal{I}_i^{\rm UL}\right)_{\rm BS}|\rho\right)$ and $\mathbb{E}\!\left(\left(\mathcal{I}_i^{\rm UL}\right)_{\rm UE}|\rho\right)$ (see Appendix \ref{Appendix4}).
\end{proof}

Regarding the LI, let $\mathcal{I}_{\rm LI}^{\rm DL}$ and $\mathcal{I}_{\rm LI}^{\rm UL}$ denote the residual interference at the FD UEs and the FD BSs, respectively, after the imperfect LI cancellation \cite{SIL2}. Since the residual LI incurred at a given receiver depends on its own transmit power, we define the residual LI power as follows
\begin{equation}\label{LI}
	\mathcal{I}_{\rm LI}^{\rm DL}\! =\! P_u(r_i\!\left(\|x_{0}\|\right))|\hat{h}^u_{\rm LI}|^2\ \text{and}\  \mathcal{I}_{\rm LI}^{\rm UL}\!=\!P|\hat{h}^b_{\rm LI}|^2,
\end{equation}
where $\hat{h}_{\rm LI}^b$ and $\hat{h}_{\rm LI}^u$ represent the estimate of the residual LI channel coefficient at the BSs and the UEs, following a Nakagami-$\mu$ distribution with parameters $\left(\mu,\sigma^2_{e_{\rm LI}^b}\right)$ and $\left(\mu,\sigma^2_{e_{\rm LI}^u}\right)$, respectively. An explicit expression for the variance of the CE error for the LI links at the BSs and the UEs i.e., $\sigma^2_{e_{\rm LI}^b}$ and $\sigma^2_{e_{\rm LI}^u}$, are given in Section \ref{VarianceCELI}.
.

\subsection{Channel Estimation}
\begin{figure}
	\centering\includegraphics[width=0.55\linewidth]{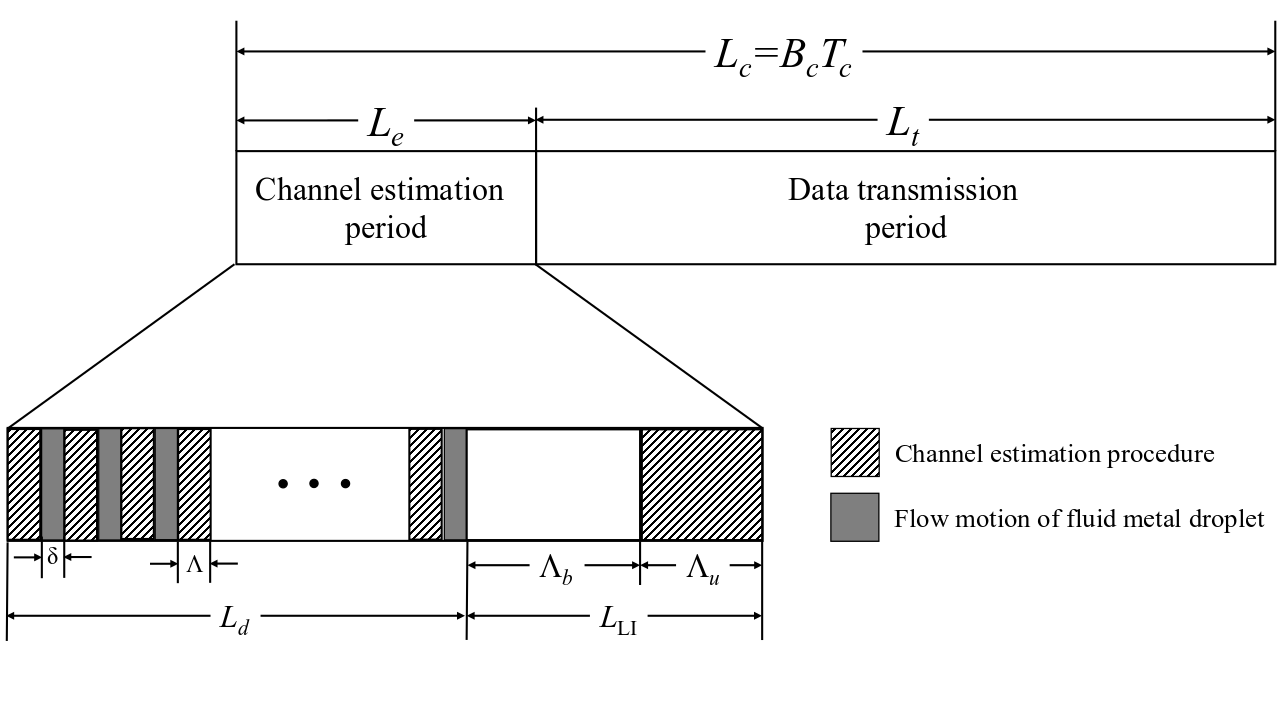}
	\caption{Representation of a block fading channel consisting of channel estimation and data transmission.}\label{Time}\vspace{-20pt}
\end{figure}
In this section, we introduce the LMMSE-based CE scheme adopted for the CE of both the direct and the LI links, for the considered limited coherence interval scenario. More specifically, we consider a low-complexity LMMSE-based CE technique, where the CE between the UEs' ports and their serving BSs, denoted as \textit{direct links}, as well as between the transmit and the receive antenna of a FD transmitter (i.e., either BS or UE), denoted as \textit{LI links}, is performed via pilot-training symbols in a sequential manner. {Due to the fact that the propagation channels are modeled by a linear system and each channel response follows a circularly symmetric complex Gaussian distribution, the employment of such CE scheme enables the modeling and analysis of the considered network deployments with low complexity.} Hence, in accordance with the adopted block fading model, the length of each coherence interval/block ($L_c$ channel uses) is equal to the product of the coherence bandwidth $B_c$ in Hz and the coherence time $T_c$ in sec \cite{Tsebook}. As illustrated in Fig. \ref{Time}, each coherence block is divided into two sub-blocks for CE and data communication. In particular, we assume that $L_e$ channel uses per block are assigned for pilot-training symbols to enable CE, of which $L_d$ and $L_{\rm LI}$ channel uses are allocated for the CE of the direct and the LI links, respectively i.e., $L_e=L_d+L_{\rm LI}$. The rest of the channel uses i.e., $L_t = L_c-L_e$, are used for data transmission. Without loss of generality, we assume the forward and reverse channels between a BS and a UE are reciprocal.
\subsubsection{Direct Links}
Based on the adopted FA model, described in Section \ref{FAmodel}, a fraction of the channel uses allocated for the first phase (i.e., CE) is dedicated to the flow motion of the fluid metal between the various ports of a FA, also known as \textit{switching channel uses} and denoted as $l_s$. This captures the finite velocity of a fluid metal within the electrolyte-filled capillary of a FA, contrary to existing works that assume high velocity fluid metals that leads to instantaneous port switching i.e., $l_s=0$. In the context of the adopted sequential CE technique, the fluid metal of a FA performs $N-1$ transitions of duration $\delta$ between adjacent ports. Hence, the required switching channel uses are equal to $l_s=(N-1)\delta B_c$. The key idea of the adopted CE process is to divide the remaining CE period i.e., $L_d-l_s$, into $N$ symmetric segments of $\Lambda=\frac{L_d-l_s}{N}$ consecutive symbols (see Fig. \ref{Time}). During each segment, the channel between a single FA port and the serving BS is estimated\footnote{More sophisticated vector-based or compressive sensing-based CE schemes can be used but this estimation process is sufficient for the purpose of this work.}. Hence, the baseband equivalent received pilot signal at the typical UE's $i$-th port, is given by
\begin{equation*}
\mathbf{y}_i = \sqrt{\Lambda P\ell\left(r_{i}(\rho)\right)}g_{0i} \mathbf{X}_0+\sum\nolimits_{\substack{t\in\mathbb{N}^+ \\ x_t\in\Phi\backslash x_0}}\sqrt{P\ell\left(r_i\left(\|x_t\|\right)\right)}g_{ti}\mathbf{X}_t+\eta_0,
\end{equation*}
where $\mathbf{X}_0$ is a deterministic $\Lambda\times 1$ training symbol vector \cite{KAY}, $\mathbf{X}_t\stackrel{d}{\sim}\mathcal{CN}\left(\mathbf{0}_{\Lambda\times 1},\mathbf{I}_{\Lambda}\right)$ is a pilot symbol vector from the interfering BS at $x_t\in\Phi\backslash x_0$, with channel $g_{ti}\stackrel{d}{\sim}\mathcal{CN}(0,1)$, and $\eta_0\stackrel{d}{\sim}\mathcal{CN}\left(\mathbf{0}_{\Lambda\times 1},N_0\mathbf{I}_{\Lambda}\right)$ is the additive white Gaussian noise (AWGN) vector. 

We consider the low-complexity LMMSE estimator, which is optimal among the class of linear estimators. Thus, according to the adopted estimator, the estimate of $g_{0i}$ conditioned on $\rho$, is given by 
\begin{equation}\label{VarianceProof}
\hat{g}_{0i}|_{\rho}=\frac{\sqrt{\Lambda P\ell\left(r_{i}(\rho)\right)}\sigma^2}{\Lambda P\ell\left(r_{i}(\rho)\right)\sigma^2+N_0+\mathbb{E}\left(\left(\mathcal{I}_i^{\rm DL}\right)_{\rm BS}|\rho\right)}\widetilde{y}_i,
\end{equation}
where $\widetilde{y}_i$ is the observation scalar signal at the $i$-th port of the typical UE i.e., $\widetilde{y}_i=\mathbf{X}_0^\dagger\mathbf{y}_i$, and $\mathbb{E}\left(\left(\mathcal{I}_i^{\rm DL}\right)_{\rm BS}|\rho\right)$ denotes the conditional mean interference, that is caused by the interfering BSs and observed at the $i$-th port of the typical UE, that is given by \eqref{BS-UE}.

The CE error can then be derived as $e_{0i} = g_{0i}-\hat{g}_{0i}|\rho$, where $e_{0i}\stackrel{d}{\sim}\mathcal{CN}\left(0,\sigma^2_{e_i}\right)$, and $\hat{g}_{0i}|\rho$ and $e_{0i}$ are uncorrelated \cite{KAY}. The following lemma evaluates the variance of the CE error, given that the distance between the typical UE and its serving BS is $\rho$.
\begin{lemma}\label{Proposition1}
	The variance of the CE error for the link between the $i$-th port of the typical UE and the tagged BS at $x_0\in\Phi$, conditioned on $\rho$, is given by
	\begin{equation}\label{NoiseVarianceEq}
	\sigma^2_{e_i}\!=\!\left(\!1\!+\!\frac{\Lambda P\ell\left(r_{i}(\rho)\right)}{N_0\!+\!\frac{2\pi\sigma^2}{a-2}\!P\lambda_br_{i}^{2-a}(\rho)}\!\right)^{-1}.
	\end{equation}
where $\Lambda=\frac{L_d-l_s}{N}$.
\end{lemma}
\begin{proof}
	See Appendix \ref{Appendix1}.
\end{proof}

\subsubsection{LI Links}\label{VarianceCELI}
During the $L_{\rm LI}$ period of the proposed CE scheme, CE processes for the LI links of both the BSs and the UEs are performed, which will be exploited for the cancellation of the LI signal. A private random training signal, known to the transmitting node only, is transmitted to estimate the respective LI channels by both BSs and UEs \cite{MUR}. Independent time periods have been utilized for transmission of the private training signal by both nodes to avoid the interference from each other, which implies that the BSs remains silent while UEs are transmitting, and vice versa. More specifically, we assume that the $L_{\rm LI}$ period is divided into two segments of $\Lambda_b = w L_{\rm LI}$ and $\Lambda_u = (1-w) L_{\rm LI}$ consecutive symbols for the LI CE at the BS and at the UE of the typical link, respectively, where $w\in(0,1]$. Thus, by employing the low-complexity LMMSE estimator, the estimates of the LI links at both the BSs and the UEs can be achieved, where the variance of their CE error is evaluated in the following Lemma.
\begin{lemma}\label{VarianceLI}
	The variance of the CE error for the LI link at the $i$-th port of the typical UE, conditioned on $\rho$, is given by
	\begin{equation}\label{NoiseVarianceEqLIu}
		\sigma^2_{e_{\rm LI}^u}\!=\!\left(\!1\!+\!\frac{\Lambda_u P_u(r_{i}(\rho))v_{\rm LI}}{N_0\!+\!\mathbb{E}\left((\mathcal{I}_i^{\rm DL})_{\rm UE}|\rho\right)}\!\right)^{-1},
	\end{equation}
and for the LI link at the typical BS, is given by
\begin{equation}\label{NoiseVarianceEqLIb}
	\sigma^2_{e_{\rm LI}^b}\!=\!\left(\!1\!+\!\frac{\Lambda_b Pv_{\rm LI}}{N_0\!+\!\mathbb{E}\left((\mathcal{I}_i^{\rm UL})_{\rm BS}|\rho\right)}\!\right)^{-1}
\end{equation}
	where $\mathbb{E}\left((\mathcal{I}_i^{\rm DL})_{\rm UE}|\rho\right)$ and $\mathbb{E}\left((\mathcal{I}_i^{\rm UL})_{\rm BS}|\rho\right)$ depict the mean of UEs-to-UEs and BSs-to-BSs interference that are given by \eqref{MeanDLUE} and \eqref{MeanULBS}, respectively, $v_{\rm LI}$ is constant representing the path-loss between the transmit and receive antenna of a FD transceiver, $\Lambda_b= w L_{\rm LI}$, and $\Lambda_u = (1-w) L_{\rm LI}$
\end{lemma}
\begin{proof}
	See Appendix \ref{ProofAppendix}
\end{proof}
\subsection{Data transmission}
Regarding the data transmission period (i.e., for the rest of the coherence block after the CE process) for the DL transmission, all BSs transmit data to their associated UEs, whose FA's location is switched to the port that is estimated to provide the strongest channel in order to have the best reception performance. Hence, the FA's location is instantly switched to the port that satisfies 
\begin{equation}\label{Association}
	\mathtt{i}=\argmax_{i\in \mathcal{N}}\{\big|\hat{g}_{0i}\big|\},
\end{equation}
where $\hat{g}_{0i}$ represents the estimated channel between the serving BS and the $i$-th port of the typical UE. For the sake of simplicity, we assume that all BSs transmit both the data and the pilot-training symbols with the same power $P$ (dBm). Thus, the received signal at the $\mathtt{i}$-th port of the typical UE during the $n$-th channel use, is given by
\begin{align}\label{DataTransmissionDL}
	d^{\rm DL}_{\mathtt{i}}[n] &= \sqrt{\ell\left(r_{\mathtt{i}}(\rho)\right)}\widehat{g}_{0\mathtt{i}} s^b_{\rm 0}[n]+\sqrt{\ell\left(r_{\mathtt{i}}(\rho)\right)}e_{0\mathtt{i}} s^b_{\rm 0}[n]\nonumber\\&+\sum\nolimits_{\substack{t\in\mathbb{N}^+ \\ x_t\in\Phi\backslash x_{0}}}\sqrt{\ell\left(r_\mathtt{i}\left(\|x_t\|\right)\right)}g_{t\mathtt{i}}s^b_t[n]+\sum\nolimits_{\substack{t\in\mathbb{N}^+ \\ y_t\in\Psi}} \sqrt{\ell\left(r_\mathtt{i}\left(\|y_t\|\right)\right)}h_{ti}s^u_t[n]\nonumber\\&+\sqrt{v_{\rm LI}}e_{\rm LI}^us^u_{0}[n]+\eta_0[n],
\end{align}
where $n\in\{L_e+1,\dots, L_c\}$, $s^b_t[n]$ and $s^u_t[n]$ represent independent Gaussian distributed data symbols from the $t$-th BS and the $t$-th interfering UE, respectively, satisfying $\mathbb{E}\left[|s^b_0[n]|^2\right]=P$, $\mathbb{E}\left[|s^b_t[n]|^2\right]=P$ and $\mathbb{E}\left[|s^u_t[n]|^2\right]=P_u(\|\overline{y_t}\|)$; $\eta_0[n]\stackrel{d}{\sim}\mathcal{CN}(0,N_0)$ is AWGN. Note that, the first term of \eqref{DataTransmissionDL} is known at the receiver, while the remaining terms are unknown and are treated as noise. Therefore, an estimate of $s_0[n]$ can be formulated as $\hat{s}_0[n]=\sqrt{\ell\left(r_{\mathtt{i}}(\rho)\right)}\frac{\left(\widehat{g}_{{\rm 0}\mathtt{i}}\right)^*}{\big|\widehat{g}_{{\rm 0}\mathtt{i}}\big|^2}d^{\rm DL}_{\mathtt{i}}[n]$, from which the SINR observed at the $\mathtt{i}$-th port for the DL transmission, denoted as $\gamma_\mathtt{i}^{\rm DL}$, can be written as
\begin{equation}\label{SINRDL}
\gamma_\mathtt{i}^{\rm DL} \!=\! \frac{\frac{P}{r^a_{\mathtt{i}}(\rho)}\left| \widehat{g}_{{\rm 0}\mathtt{i}}\right|^2}{\mathcal{I}^{\rm DL}_\mathtt{i}\!+\!\Sigma^{\rm DL}_i}, 
\end{equation}
where $\mathcal{I}^{\rm DL}_\mathtt{i}$ and $\mathcal{I}_{\rm LI}^{\rm DL}$ represent the aggregate interference and LI for the DL transmission, respectively; $\sigma^2_e$ and $\sigma_{e^u_{\rm LI}}^2$ depict the variance of the CE for the direct and the UEs' LI link, respectively, and $\Sigma^{\rm DL}_i$ is the aggregate noise observed for the DL transmission at the $i$-th port, and is given by $\Sigma^{\rm DL}_i = \sigma_{e^u_{\rm LI}}^2\mathcal{I}_{\rm LI}^{\rm DL}\!+\!\frac{P}{r^a_{\mathtt{i}}(\rho)}\sigma^2_{e_\mathtt{i}}\!+\!N_0$.

In addition, during the data transmission period, all UEs transmit data to their associated BSs. We assume that both the data and the training-pilot symbols are transmitted with power $P_u(R)$ dBm, according to the adopted power control scheme. Therefore, for the typical BS that is associated with the $\mathtt{i}$-th port of its serving UE, the received signal during the $n$-th channel use, is given by 
\begin{align}\label{DataTransmissionUL}
	d^{\rm UL}_{\mathtt{i}}[n] &= \sqrt{\ell\left(r_{\mathtt{i}}(\rho)\right)}\widehat{g}_{0\mathtt{i}} s^u_{\rm 0}[n]+\sqrt{\ell\left(r_{\mathtt{i}}(\rho)\right)}e_{0\mathtt{i}} s^u_{\rm 0}[n]\nonumber\\&+\sum\nolimits_{\substack{t\in\mathbb{N}^+ \\ x_t\in\Phi}}\sqrt{\ell\left(r_\mathtt{i}\left(\|x_t\|\right)\right)}g_{t\mathtt{i}}s^b_t[n]+\sum\nolimits_{\substack{t\in\mathbb{N}^+ \\ y_t\in\Psi\backslash x_0}} \sqrt{\ell\left(r_\mathtt{i}\left(\|y_t\|\right)\right)}h_{ti}s^u_t[n]\nonumber\\&+\sqrt{v_{\rm LI}}e^b_{\rm LI}s^b_{0}[n]+\eta_0[n],
\end{align}
where $n\in\{L_e+1,\dots, L_c\}$, $s^u_0[n]$ represents independent Gaussian distributed data symbols from the typical UE, satisfying $\mathbb{E}\left[|\bar{s}_0[n]|^2\right]=P_u(\ell(r_i(\rho)))$. Similar as before, the first term of \eqref{DataTransmissionUL} is known at the receiver, while the remaining terms are unknown and are treated as additive noise. By following similar methodology as for the DL transmission, the SINR observed at the typical BS for the UL transmission, denoted as $\gamma_\mathtt{i}^{\rm UL}$, can be written as
\begin{equation}\label{SINRUL}
	\gamma_\mathtt{i}^{\rm UL} \!=\! \frac{\frac{P_u(\ell(r_i(\rho)))}{r^a_{\mathtt{i}}(\rho)}\left| \widehat{g}_{{\rm 0}\mathtt{i}}\right|^2}{\mathcal{I}^{\rm UL}_\mathtt{i}\!+\!\Sigma^{\rm UL}_i}, 
\end{equation}
where $\mathcal{I}^{\rm UL}_\mathtt{i}$ and $\mathcal{I}_{\rm LI}^{\rm UL}$ represent the aggregate interference and LI for the UL transmission, respectively; $\sigma_{e^b_{\rm LI}}^2$ depict the variance of the CE for the BSs' LI link and $\Sigma^{\rm UL}_i$ is the aggregate noise observed by the typical BS that is served by the $i$-th port of the typical UE for the UL transmission, and is given by $\Sigma^{\rm UL}_i = \sigma_{e^b_{\rm LI}}^2\mathcal{I}_{\rm LI}^{\rm UL}\!+\!\frac{P_u(\ell(r_i(\rho)))}{r^a_{\mathtt{i}}(\rho)}\sigma^2_{e_\mathtt{i}}\!+\!N_0$.

\section{FA-aided FD Cellular Networks}\label{PerformanceAnalysis}
In this section, we analytically derive both the DL and the UL outage probabilities, as well as the sum-rate performance of a homogeneous FA-aided FD cellular network. Initially, we assess the statistical properties of the estimated channels under the adopted LMMSE-based CE technique, providing analytical expressions for the joint pdf and cumulative distribution function (cdf) of the estimated channels. Then, the SINR cdf is defined for both the DL and the UL transmission. Finally, by leveraging the derived SINR distribution framework, analytical expressions for the sum-rate performance in the context of the presented homogeneous FA-aided FD cellular network are obtained.

\subsection{Preliminary Results}
Firstly, the joint pdf and cdf of $|\hat{g}_{{\rm 0}1}|,\dots,|\hat{g}_{{\rm 0}N}|$, conditioned on $\rho$, are given in the following lemma.
\begin{lemma}\label{Lemma1}
	The joint cdf and pdf of $\left| \widehat{g}_{{\rm 0}1}\right|,\dots,\left|\widehat{g}_{{\rm 0}N}\right|$, conditioned on $\rho$, are given by
	\begin{equation}\label{CDF}
		F_{\left| \widehat{g}_{{\rm 0}1}\right|,\dots,\left|\widehat{g}_{{\rm 0}N}\right|}(\tau_1,\dots,\tau_{N}|\rho)=\!\int_0^{\frac{\tau_1^2}{\widetilde{\sigma}_1^2}}\!\exp\!\left(-t\right)\!\prod_{j\in\mathcal{N}}\!\left[\!1\!-\!Q_1\left(\sqrt{2\mu_j^2\frac{\widetilde{\sigma}_1^2}{\widetilde{\sigma}_j^2}t},\sqrt{\frac{2}{\widetilde{\sigma}_j^2}}\tau_j\right)\right]{\rm d}t
	\end{equation}
	and
	\begin{equation}\label{PDF}
		f_{\left| \widehat{g}_{{\rm 0}1}\right|,\dots,\left|\widehat{g}_{{\rm 0}N}\right|}(\tau_1,\dots,\tau_{N}|\rho)=\prod_{\substack{i\in\mathcal{N}\\ (\mu_1\triangleq 0)}}\frac{2\tau_i}{\widetilde{\sigma}_i^2}\exp\left(-\frac{\tau_i^2+\mu^2_i\tau_1^2}{\widetilde{\sigma}_i^2}\right)I_0\left(\frac{2\mu_i\tau_1\tau_i}{\widetilde{\sigma}_i^2}\right),
	\end{equation}
	respectively, where $\tau_1,\dots,\tau_{N}\geq 0$, $I_0(\cdot)$ depicts the zero-order modified Bessel function of the first kind, $Q_1(\cdot,\cdot)$ is the first-order Marcum Q-function, and $\widetilde{\sigma}_i^2=\sigma^2(1-\mu_i^2)+\sigma^2_{e_i}|_{\rho}$.
\end{lemma}
\begin{proof}
	See Appendix \ref{Appendix2}.
\end{proof}
The overall performance of FD transceivers in the context of large-scale multi-cell networks is jeopardized by the increase of both the intra- and inter- cell co-channel interference \cite{SGbook}. According to the considered network deployment, the overall interference observed by the $i$-th port of the typical UE and the typical BS, is given by \eqref{OverallDL} and \eqref{OverallUL}, respectively.  Even though the performance of a communication network by considering the actual multi-user interference is analytically tractable for the PPP case with independent fading channels, in most relevant (realistic) models, it is either impossible to analytically analyze or cumbersome to evaluate even numerically. Motivated by the aforementioned discussion, the following proposition presents our assumption to approximate the multi-user interference distribution of large-scale wireless networks by using Gamma distribution, aiming to provide simple and tractable expressions for the outage and sum-rate performance.
\begin{proposition}\label{InterferenceDistribution}
	The multi-user interference observed by the typical FD receiver (i.e., BS or UE), conditioned on the distance its serving transmitter i.e., $\rho$, follows a Gamma distribution with pdf
	\begin{equation}\label{Gammapdf}
		f_{\mathcal{I}}(x|\rho)=\frac{x^{\varpi-1}\exp\left(-\frac{x}{\varrho}\right)}{\Gamma[\varpi]\varrho^\varpi},\ x>0,
	\end{equation}
	with shape parameter $\varpi=(\mathbb{E}[\mathcal{I}])^2/{\rm Var}(\mathcal{I})$ and scale parameter $\varrho={\rm Var}(\mathcal{I})/\mathbb{E}[\mathcal{I}]$.
\end{proposition}

\subsection{Outage Performance}\label{SINRstatistics}
We evaluate the outage probability of the typical receiver (i.e., either a BS or a UE) for both the DL and UL transmissions of homogeneous FA-aided FD cellular networks. The outage probability for both the DL and UL transmissions, $\mathcal{P}_o^{\rm DL}(\theta)$ and $\mathcal{P}_o^{\rm UL}(\theta)$, respectively, is the probability that the SINR is less than a threshold $\theta$. Therefore, we can mathematically describe the above-mentioned DL and UL outage performance with the probabilities $\mathcal{P}_o^{\rm DL}(\theta)= \mathbb{P}\left[\gamma_\mathtt{i}^{\rm DL}<\theta\right]$ and $\mathcal{P}_o^{\rm UL}(\theta)= \mathbb{P}\left[\gamma_\mathtt{i}^{\rm UL}<\theta\right]$, respectively, which are analytical evaluated in the following theorems.

\begin{theorem}\label{Theorem1}
The outage probability achieved by the typical receiver for the $\varUpsilon$ transmission, where $\varUpsilon\in\{{\rm DL,UL}\}$, in the context of FA-aided FD cellular networks, is given by 
\begin{equation}\label{DLOutage}
	\mathcal{P}_o^\varUpsilon(\theta)= \prod_{i\in\mathcal{N}}\int_{0}^\infty \mathcal{P}_o^\varUpsilon(\theta|\rho)2\pi\lambda_b \rho\exp\left(-\pi\lambda_b  \rho^2\right){\rm d}\rho,
\end{equation}
where 
\begin{align}\label{Conditional}
\mathcal{P}_o^\varUpsilon(\theta|\rho) &\!=\! \int\nolimits_0^\infty\!\Bigg(\int\nolimits_0^{\frac{\Theta^2_\varUpsilon(x_1)}{\widetilde{\sigma}_1^2}}\exp\left(-t\right)\nonumber\\&\times\prod_{j\in\mathcal{N}}\!\left[\int\nolimits_0^\infty\!\left(1\!-\!Q_1\!\left(\sqrt{2\mu_j^2\frac{\widetilde{\sigma}_1^2}{\widetilde{\sigma}_j^2}t},\sqrt{\frac{2}{\widetilde{\sigma}_j^2}}\Theta_\varUpsilon(x_j)\right)f_{\mathcal{I}}(x_j|\rho){\rm d}x_j\right)\right]{\rm d}t\Bigg)f_{\mathcal{I}}(x_1|\rho){\rm d}x_1,
\end{align}
and
\begin{equation}\label{Threshold}
	\Theta_\varUpsilon(i)=\begin{cases}
		\frac{\theta r_i^{a}(\rho)}{P}\left(\mathcal{I}^{\rm DL}_i\!+\!\Sigma^{\rm DL}_i\right),&\text{if $\varUpsilon=DL$},\\
		\frac{\theta r_i^{a}(\rho)}{P_u(\ell(r_i(\rho)))}\left(\mathcal{I}^{\rm UL}_i\!+\!\Sigma^{\rm UL}_i\right),&\text{if $\varUpsilon=UL$},
	\end{cases}
\end{equation}
$\Sigma^{\rm DL}_i=\sigma_{e^u_{\rm LI}}^2\mathcal{I}_{\rm LI}^{\rm DL}+\frac{P}{r^a_{{i}}(\rho)}\sigma^2_{e_{i}}+N_0$, $\Sigma^{\rm UL}_i=\sigma_{e^b_{\rm LI}}^2\mathcal{I}_{\rm LI}^{\rm UL}+\frac{P_u(\ell(r_i(\rho)))}{r^a_{{i}}(\rho)}\sigma^2_{e_{i}}+N_0$,  $\mathcal{P}_o^{\varUpsilon}(\theta|\rho)$ represents the conditional outage probability for the $\varUpsilon$ transmission, that is given by \eqref{Conditional}; $\mathcal{I}_{\rm LI}^{\rm DL}$ and $\mathcal{I}_{\rm LI}^{\rm UL}$ depict the residual LI observed at the $i$-th port of the typical UE and the typical BS, respectively, that are given by \eqref{LI}, and $\widetilde{\sigma}_j^2=\sigma^2(1-\mu_j^2)+\sigma^2_{e_i}$.
\end{theorem} 
\begin{proof}
	See Appendix \ref{Appendix3}
\end{proof}

{In spite of the fact that Theorem \ref{Theorem1} provides an analytical approach to obtain the outage probability of the considered FA-aided FD communications, the analysis of the achieved performance is still cumbersome and tedious, impeding the extraction of meaningful insights. To this end, we evaluate the achieved performance in the asymptotic regime. More specifically, by considering the special case where the multi-user interference of large-scale wireless networks is approximated by its mean value i.e., $\mathcal{I}_j^\varUpsilon = \mathbb{E}\left[\mathcal{I}_j^\varUpsilon\right]$, with $j\in\mathcal{N}$ and $\varUpsilon\in\{{\rm DL, UL}\}$, in the following lemma, an approximation for the conditional outage probability i.e., $\mathcal{P}_o^\varUpsilon(\theta|\rho)$, can be derived.

\begin{lemma}
The conditional outage probability achieved by the typical receiver for the $\varUpsilon$ transmission, where $\varUpsilon\in\{{\rm DL,UL}\}$, in the context of FA-aided FD cellular networks, can be approximated by
\begin{equation}\label{ConditionalA}
\mathcal{P}_o^\varUpsilon(\theta|\rho) \!\simeq\! \int\nolimits_0^{\frac{\Theta^2_\varUpsilon(1)}{\widetilde{\sigma}_1^2}}\exp\left(-t\right)\prod_{j\in\mathcal{N}}\!\left[1\!-\!Q_1\!\left(\sqrt{2\mu_j^2\frac{\widetilde{\sigma}_1^2}{\widetilde{\sigma}_j^2}t},\sqrt{\frac{2}{\widetilde{\sigma}_j^2}}\Theta_\varUpsilon(j)\right)\right]{\rm d}t,
\end{equation}
where 
\begin{equation}\label{ThresholdA}
\Theta_\varUpsilon(i)=\begin{cases}
\frac{\theta r_i^{a}(\rho)}{P}\left(\mathbb{E}[\mathcal{I}^{\rm DL}_i]\!+\!\Sigma^{\rm DL}_i\right),&\text{if $\varUpsilon=DL$},\\
\frac{\theta r_i^{a}(\rho)}{P_u(\ell(r_i(\rho)))}\left(\mathbb{E}[\mathcal{I}^{\rm UL}_i]\!+\!\Sigma^{\rm UL}_i\right),&\text{if $\varUpsilon=UL$},
\end{cases}
\end{equation}
and $\mathbb{E}[\mathcal{I}^\varUpsilon_i]$ represents the mean interference observed at the $i$-th port of the FA for the $\varUpsilon$ direction, where $\varUpsilon\in\{{\rm DL,UL}\}$ and $i\in\mathcal{N}$.
\end{lemma}
\begin{proof}
	By substituting the random variables $\mathcal{I}^{\rm DL}_i$ and $\mathcal{I}^{\rm UL}_i$ with their mean values i.e., $\mathbb{E}[\mathcal{I}^{\rm DL}_i]$ and $\mathbb{E}[\mathcal{I}^{\rm UL}_i]$, respectively, the final expression can be derived.
\end{proof}}
\subsection{Sum-Rate Performance}
Another extremely important performance metric is the average sum-rate performance (bits/sec) that indicates the information rate that can be transmitted over a given bandwidth for the considered network deployment. More specifically, the average sum-rate performance measures the quantity of UEs that can be simultaneously supported by a limited RF bandwidth in a defined geographic area. We can mathematically describe the average sum-rate performance, denoted as $\mathcal{R}$, by the following expression
\begin{equation*}
	\mathcal{R} = B_c\left(1-\frac{L_e}{L_c}\right)\left(\log\left(1+\gamma_{\mathtt{i}}^{\rm DL}\right)+\log\left(1+\gamma_{\mathtt{i}}^{\rm UL}\right)\right),
\end{equation*}
where $B_c$ is the coherence bandwidth, $\left(1-\frac{L_e}{L_c}\right)$ depicts the fractional amount of time (relative to the total frame length) used for data transmission; $\gamma_{\mathtt{i}}^{\rm DL}$ and $\gamma_{\mathtt{i}}^{\rm UL}$ are the SINR observed at the typical receiver for the DL and the UL transmission, respectively. In the following theorem, analytical expressions for the average sum-rate performance of the typical FD link in the context of the considered homogeneous FA-aided FD cellular networks are provided, by exploiting the SINR distribution framework that is obtained in Section \ref{SINRstatistics}.
\begin{theorem}
	The average sum-rate performance in the context of FA-aided FD cellular networks, can be expressed as
	\begin{equation}\label{FDrate}
		\mathcal{R} = \frac{B_c\left(1-\frac{L_t}{L_c}\right)}{\ln(2)}\int_0^\infty\frac{\overline{\mathcal{P}}_o^{\rm DL}(\theta)+\overline{\mathcal{P}}_o^{\rm UL}(\theta)}{\tau+1}\diff \theta,
	\end{equation}
	where $\overline{\mathcal{P}}_o^{\rm DL}(\theta)$ and $\overline{\mathcal{P}}_o^{\rm DL}(\theta)$ express the coverage probability for the DL and the UL transmissions, respectively i.e., $\overline{\mathcal{P}}_o^{\rm DL}(\theta)=1-{\mathcal{P}}_o^{\rm DL}(\theta)$ and $\overline{\mathcal{P}}_o^{\rm UL}(\theta)=1-{\mathcal{P}}_o^{\rm UL}(\theta)$.
\end{theorem}
\begin{table*}[t]\centering
	\caption{Simulation Parameters.}\label{Table2}
	\scalebox{0.85}{
		\begin{tabular}{| l | l || l | l |}\hline
			\textbf{Parameter} 		& \textbf{Value} 									& \textbf{Parameter} 						& \textbf{Value}\\\hline
			BSs' density ($\lambda_b$) & $5\times10^{-5}$	&Dimension ratio ($D/L$) & $5$  \\\hline
			Number of ports ($N$) & $15$  	&  Receive sensitivity ($\omega$) & $-40$ dB  \\\hline
			Path-loss exponent ($a$) 	& $4$ 	&  Channel variance ($\sigma$) & $1$ \\\hline
			Scaling constant ($\kappa$)	& $0.2$ 	&  UEs' power constrain ($P_{\rm m}$) & $30$ dBm\\\hline
			Wavelength ($\lambda$) 	& $0.06$ cm 	&  Bandwidth ($W_c$) 	& $100$ MHz  \\\hline
			Initial charge ($q$) 	& $0.07$ V  	&  Coherence time ($T_c$)	& $50$ ms   \\\hline
			Viscosity ($\mu$) 	& $0.002$  &  Noise Variance ($N_o$)	& $10^{-5}$\\\hline
			Voltage difference ($\Delta\phi$) & $10$ V & LI path loss ($v_{\rm LI}$) & 0.001\\\hline
	\end{tabular}\vspace{-20pt}}
\end{table*}
\begin{figure*}[t]
	\begin{subfigure}{.49\textwidth}\centering
		\includegraphics[width=0.9\linewidth]{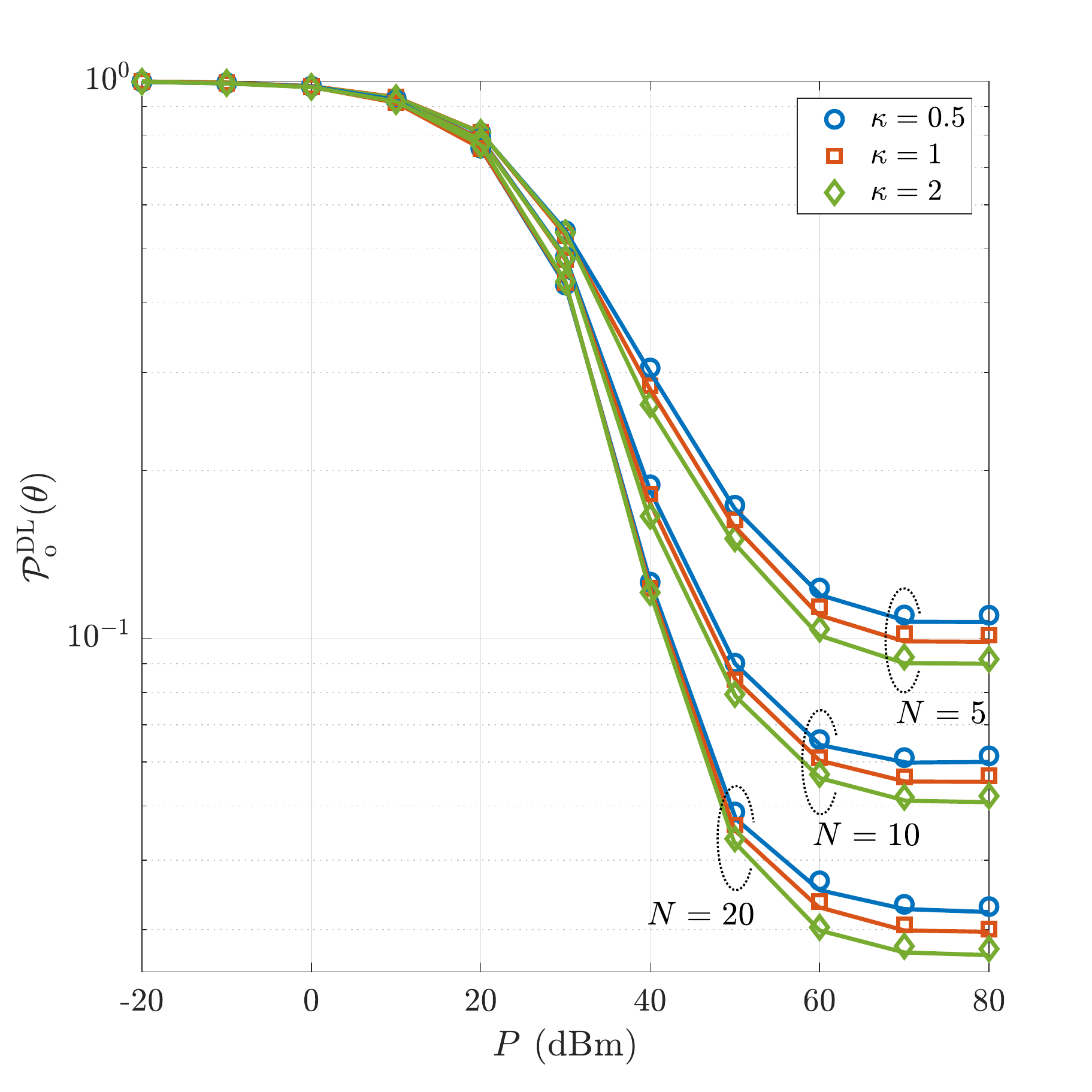}
		\caption{DL outage probability versus the transmit power ($P$).}\label{fig:fig1}
	\end{subfigure}
	\begin{subfigure}{.49\textwidth}\centering
		\includegraphics[width=0.9\linewidth]{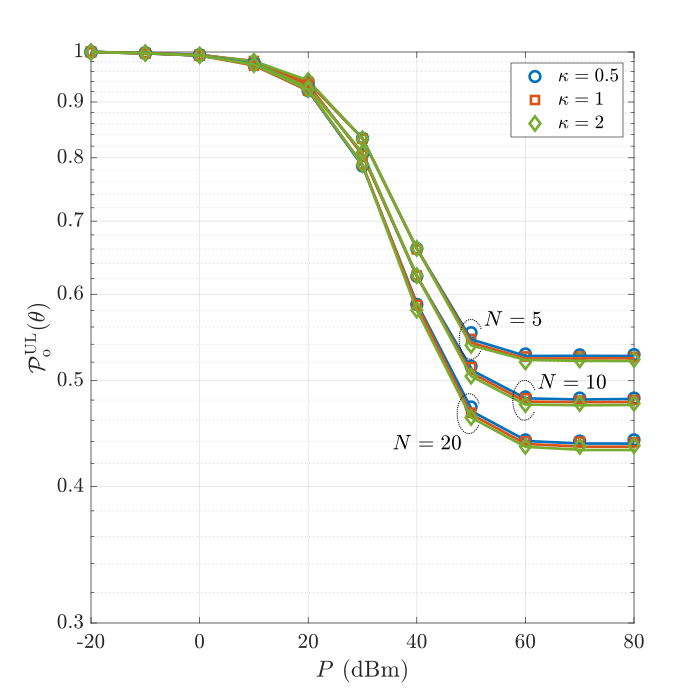}
		\caption{UL outage probability versus the transmit power ($P$).}\label{fig:fig2}
	\end{subfigure}
	\caption{Outage performance for the DL and UL transmission for different $N$ and $\kappa$; $\theta=-20$ dB, $\epsilon = 0.8$.}\label{fig}\vspace{-5mm}
\end{figure*}
\begin{proof}
	Based on the definition, $\mathcal{R}$ is given by
	\begin{align}
		\mathcal{R} &= \mathbb{E}\left[B_c\left(1-\frac{L_e}{L_c}\right)\left(\log\left(1+\gamma_{\mathtt{i}}^{\rm DL}\right)+\log\left(1+\gamma_{\mathtt{i}}^{\rm UL}\right)\right)\right]\nonumber\\
		&=\!B_c\!\left(\!1\!-\!\frac{L_e}{L_c}\right)\!\!\int_0^\infty\!\!\!\big(\mathbb{P}\left[\gamma_{\mathtt{i}}^{\rm DL}\!>\!2^{t}\!-\!1\right]\!+\!\mathbb{P}\left[\gamma_{\mathtt{i}}^{\rm UL}\!>\!2^{t}\!-\!1\right]\big)\!\diff t\nonumber.
	\end{align}
Then, by changing variable for the integral i.e., $\theta=2^t-1$, the final expression can be derived.
\end{proof}

\section{Numerical Results}\label{Numerical}
In this section, we provide numerical results to verify our model and illustrate the performance of large-scale FA-aided FD communications. A summary of the model parameters is given in Table \ref{Table2}. Please note that, the selection of the simulation parameters is made for the sake of the presentation. The use of different values leads to a shifted network performance, but with the same observations and conclusions.

Fig. \ref{fig:fig1} and Fig. \ref{fig:fig2} reveal the effect of the transmit power on the achieved outage performance for both the DL and UL transmissions, respectively. More specifically, we plot the outage probabilities, $\mathcal{P}_{\rm o}^{\rm DL}$ and $\mathcal{P}_{\rm o}^{\rm UL}$, with respect to the transmit power $P$ (dBm), for different number of FAs' ports i.e., $N=\{5,10,20\}$ and scaling constants $\kappa = \{0.5,1,2\}$. We can easily observe that, the number of FAs' ports causes the reduction of the outage performance experienced by a UE in the considered network deployment. This observation was expected since, a higher diversity gain can be attained with the increased number of FA ports, leading to an enhanced SINR observed by the UEs. In addition, it is clear from the figure that the outage probability asymptotically converges to a constant value, since as the transmission power of the nodes increases, the additive noise in the network becomes negligible. Furthermore, Fig. \ref{fig:fig1} and Fig. \ref{fig:fig2} show that larger FA architectures i.e., a larger $\kappa$, result to a reduced outage performance. As expected, by increasing the size of FAs, the distance between their ports also increases, limiting the negative effect of the spatial correlation between the ports’ channels on the network performance. Finally, the agreement between the theoretical curves (solid lines) and the simulation results (markers) validates our mathematical analysis.

\begin{figure}
	\centering\includegraphics[width=0.55\linewidth]{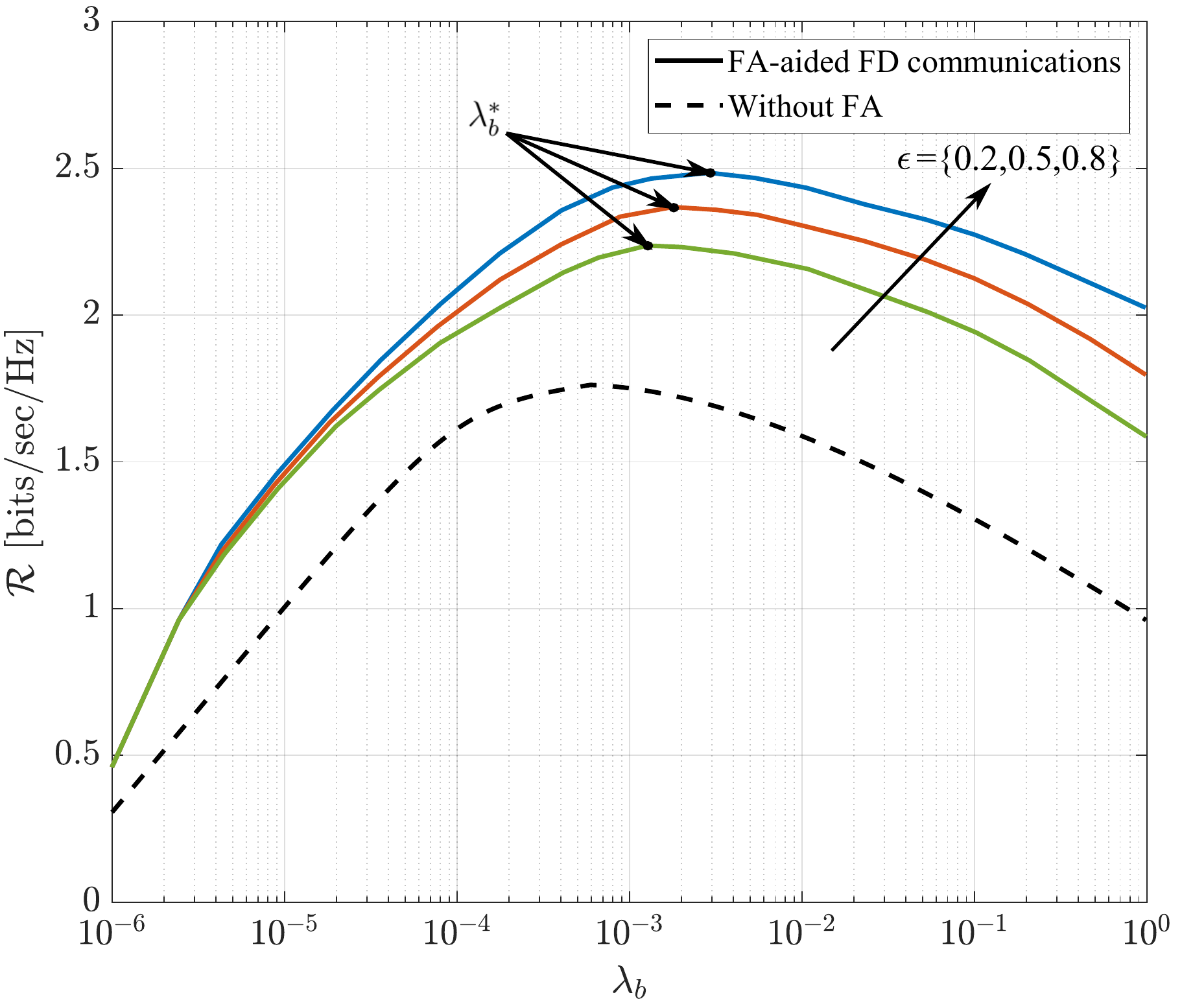}
	\caption{Average Sum-Rate performance versus the spatial density of BSs $\lambda_b$ for different $\epsilon$.}\label{fig3}\vspace{-15pt}
\end{figure}
Fig. \ref{fig3} illustrates the impact of both spatial density and power control on the average sum-rate performance for the considered FA-aided FD communications. In particular, we plot the average sum-rate performance, $\mathcal{R}$, versus the spatial density of BSs, $\lambda_b$, for different power control factors $\epsilon=\{0.2,0.5,0.8\}$. Initially, an interesting observation is that the sum-rate performance initially increases with $\lambda_b$ but, beyond a critical point, i.e. $\lambda_b^*$, it starts to decrease. This was expected since, at low density values, the increased number of BSs results in the reduction of the distance between the UEs and their serving BSs, thereby the observed SINR at both the BSs and the UEs is enhanced. Nevertheless, by further increasing the spatial density of the BSs, more and more links between BSs and UEs are activated, leading to a significant increase of the overall observed interference, and consequently the reduced ability of the network's receivers to decode the received signal. In addition, Fig. \ref{fig3} reveals the positive effect of the power control on the network's average sum-rate performance. This was expected since, the increased ability of UEs to compensate the path-loss by increasing the power control factor $\epsilon$, leading to a reduced UL outage performance and consequently to a significantly improved average sum-rate performance. {For comparison purposes, we also present the performance achieved by the conventional FD communications (i.e., $N=1$), where both BSs and UEs operate in FD mode and are equipped with a single omnidirectional antenna, denoted as ``Without FAs''.} We can easily observe that, the employment of FAs in the context of FD communications elevates the spatial diversity gain, leading to an improved average sum-rate performance with respect to the non-FA networks counterpart. {We also numerically investigate the critical spatial density that maximizes the achieved sum-rate performance in terms of the different power control factors, i.e. $\lambda_b^*=\{1.230,1.851,2.974\}\times10^{-3}$ for $\epsilon=\{0.2,0.5,0.8\}$, respectively.}

\begin{figure}
	\centering\includegraphics[width=.55\linewidth]{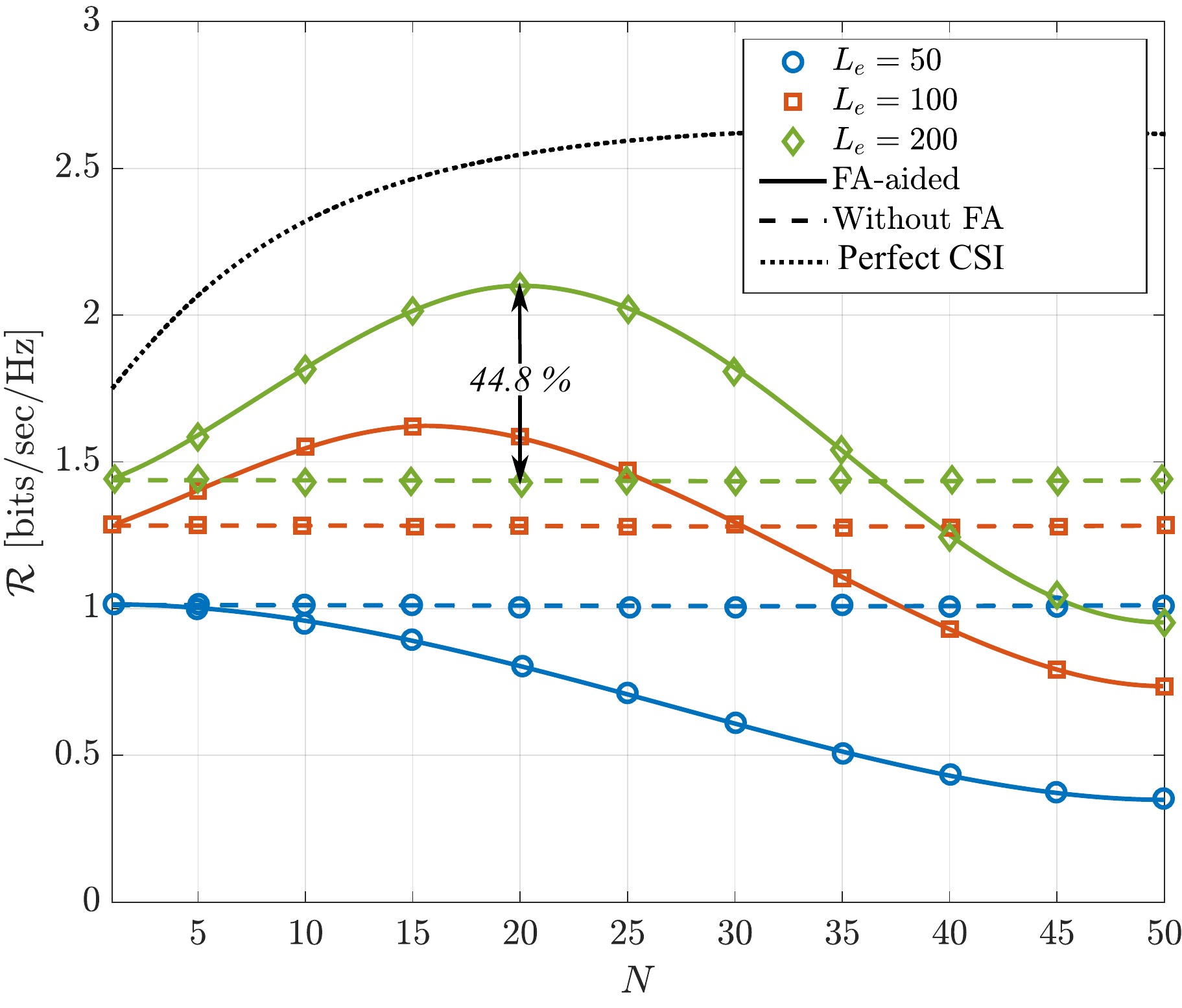}
	\caption{Average Sum-Rate performance versus the number of ports $N$ for different $L_e$.}\label{fig4}\vspace{-15pt}
\end{figure}
Fig. \ref{fig4} demonstrates the impact of the number of FAs' ports on the average sum-rate performance achieved by FA-aided FD communications. More specifically, we plot the average sum-rate $\mathcal{R}$ achieved by FA-aided FD communications as well as by the conventional FD communications without FAs, versus the number of FAs' ports, $N$, for different number of pilot-training symbols $L_e = \{50,100,200\}$. The first main observation is that for a small number of pilot-training symbols e.g., $L_e=50$ symbols, the presence of additional ports results in a reduced FA-aided FD network's performance. This is due to the fact that, under the considered limited coherence interval scenario, the number of pilot-training symbols dedicated for the CE of each port decreases, thereby the quality of the CE is reduced (i.e., $\sigma_{e_i}\rightarrow 1$), jeopardizing the achieved network performance. On the other side, due to the existence of a single antenna element in the context of conventional FD communications, the allocated training symbols are sufficient to achieve satisfactory CE quality, achieving better network performance compared to that of the FA-aided FD communications. Nevertheless, by further increasing the number of pilot-training symbols e.g., $L_e=\{100,200\}$, a sufficient number of training symbols can be allocated for the CE of all FA's ports. Therefore, by increasing the number of FAs' ports, the average sum-rate performance achieved by the considered network deployments improves. This observation can be explained by the fact that, a higher receive diversity gain can be achieved with the increased number of FA ports, and therefore, enhanced DL and UL SINR are observed. However, by further increasing the number of FAs' ports beyond a critical point, the network performance reduces. This is justified by the fact that, for a large number of FAs' ports, the allocated number of training-pilot symbols for the estimation of the channel of each port is decreased, thereby the quality of the CE is diminished, alleviating the achieved network performance. Thus, for a sufficient pilot-training symbols, the synergy of FA technology and FD radio is beneficial, providing an increase of the average sum-rate performance by around 45\% compared with the conventional FD communications, for the scenario where all FAs are equipped with $N = 20$ ports. For comparison purposes, we also present the outage performance obtained with a perfect (a-priori) CSI, denoted as ``Perfect CSI''. We can easily observe that, in contrast to the scenario considered with CE error, the network performance with a perfect CSI is constantly increasing with the increase of FA ports, posing an upper bound for the performance achieved by FA-aided FD communications. This is due to the fact that the negative effect of channel estimation quality on the network's performance is neglected.

Fig. \ref{fig5} shows the trade-off between the average sum-rate performance and the DL outage performance with respect to the number of ports, i.e $N=\{1,5,10,\dots,50\}$ and voltage gradient $\Delta\phi=\{1,10,100\}$ Volts (V). Each point in the curves represents the trade-off between the two performance metrics for a given number of ports. As mentioned before, the experienced outage performance of a FA-based UE reduces with the increase of the number of ports. In addition, Fig. \ref{fig5} illustrates that the average sum-rate performance initially increases with $N$ but, after a certain value of $N$, it starts to increase. Fig. \ref{fig5} also illustrates the impact of the considered MEMS (i.e., the system responsible for the flow movement of the liquid radiating element of FAs) on the achieved network performance. In particular, by increasing the voltage gradient ($\Delta\phi$) along the FA, the DL outage performance decreases while the average sum-rate performance increases. This was expected as, by applying a higher voltage gradient across the FA, the liquid radiating element can move at a higher speed requiring fewer switching channel uses ($l_s$), and thus, owing to the allocation of more channel uses for the CE of the FAs' links, the CE quality as well as the observed SINR are enhanced. An interesting observation is that, the optimal number of FAs' ports that maximizes the average sum-rate performance increases with the increase of the applied voltage gradient. This can be explained by the fact that, by increasing the voltage gradient along the FA, the velocity of fluid metal increases while the incurred transition delay reduces, enabling the sufficient CE of more FAs' ports, leading to an improved network performance. The same behavior is also observed for the trade-off between the average sum-rate performance and the UL outage probability.
\begin{figure}
	\centering\includegraphics[width=.55\linewidth]{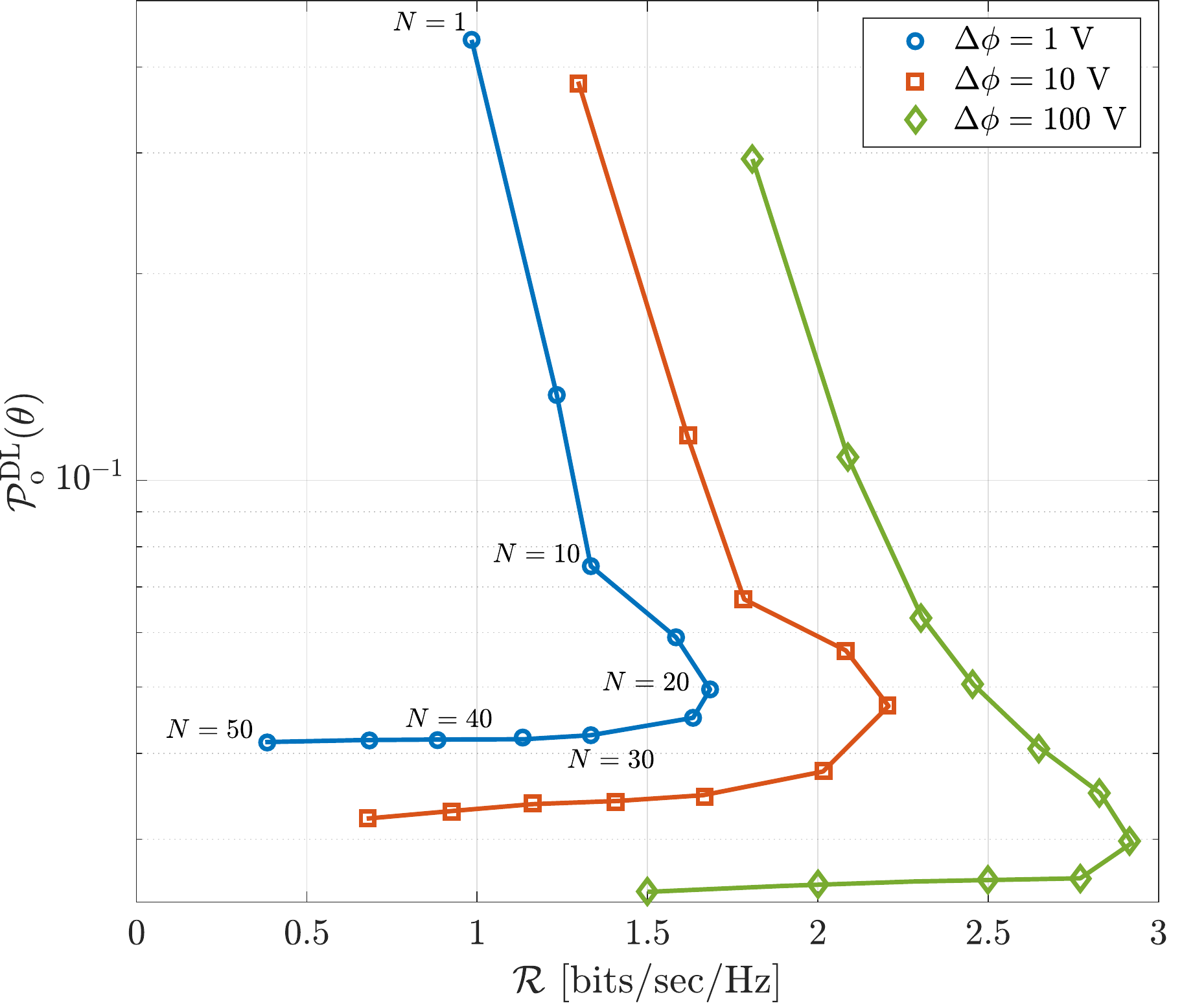}
	\caption{DL outage versus average sum-rate performance for different $N$ and $\Delta\phi$; $\theta=-20$ dB.}\label{fig5}\vspace{-15pt}
\end{figure}
\section{Conclusion}\label{Conclusion}
In this paper, we presented an analytical framework based on stochastic geometry and studied the performance of large-scale FA-aided FD cellular networks. The developed mathematical framework takes into account the employment of a LMMSE-based CE technique by all network's nodes, and also captures the presence of both CE error and channel correlation effects. We derive analytical expressions for both the outage and the average sum-rate performance, and the impact of nodes density, power control, block length, and number of FAs' ports has been discussed. Our results highlight the impact of the FAs' architecture and the network topology on the optimal number of FA and ports, providing guidance for the planning of FA-aided FD cellular networks in order to achieve enhanced network performance. Finally, we have shown that the combination of FD radio with FAs provides an increase of around $45\%$ in terms of average sum-rate performance achieved by the cellular networks, compared to that achieved by the conventional static FD communications. A future extension of this work is the consideration of FA technology in the context of MIMO systems.
\appendices
\section{Proof of Lemma \ref{MeanInterferenceDL}}\label{Appendix4}
The mean of the random variable $\left(I^{\rm UL}\right)_{\rm BS}$ conditioned on the random distance $\rho$ between the serving UE and the typical BS, which we denote as $\mathbb{E}\left(\left(I^{\rm UL}\right)_{\rm BS}|\rho\right)$, can be evaluated with the use of probability generating functionals (PGFLs). We first calculate the Laplace transform of the $\left(I^{\rm UL}\right)_{\rm BS}$ as follows
\begin{align}
	\mathcal{L}_{\left(I^{\rm UL}\right)_{\rm BS}}(s) &= \mathbb{E}\left[\exp\left(-s\left(I^{\rm UL}\right)_{\rm BS}\right)\right]\nonumber\\&=\mathbb{E}\left[\exp\left(-sP\sum\nolimits_{\substack{t\in\mathbb{N}^+ \\ x_t\in\Phi\backslash x_0}} \ell(r_{i}(\|x_t\|))|g_{ti}|^2\right)\right]\label{Proof41}\\
	&=\exp\left(-\int_{r_{i}(\rho)}^\infty\left(1-\frac{1}{1+sPx^{-a}}\right)\lambda_b\diff x\right)\label{Proof43}
\end{align}
where \eqref{Proof41} use the definition of $\left(I^{\rm UL}\right)_{\rm BS}$; \eqref{Proof43} is due to the probability generating functional for a PPP with spatial density $f_i(r)$ and by assuming that the distance between the $i$-th port of the typical UE and the the interferer is dominated by their spatial distance i.e., $r_{i}(\|x_t\|)\approxeq \|x_t\|$. The $n$-th moment of the interference power $\mathcal{I}$, can be calculated as
\begin{equation*}
	(-1)^n\frac{\diff^n}{\diff s^n}\mathcal{L}_{\mathcal{I}}(s)\Big|_{s=0}.
\end{equation*}
Thus, the mean and the variance of the $\left(I^{\rm UL}\right)_{\rm BS}$ can be evaluated as 
\begin{equation*}
	\mathbb{E}\left[\left(I^{\rm UL}\right)_{\rm BS}\right]=-\frac{\diff}{\diff s}\mathcal{L}_{\left(I^{\rm UL}\right)_{\rm BS}}(s)\Big|_{s=0}
\end{equation*}
and
\begin{equation*}
	{\rm var}\!\left[\!\left(I^{\rm UL}\right)_{\rm BS}\!\right]\!=\!\frac{\diff^2}{\diff s^2}\mathcal{L}_{\left(I^{\rm UL}\right)_{\rm BS}}(s)\Big|_{s=0}\!+\!\left(\!\frac{\diff}{\diff s}\mathcal{L}_{\left(I^{\rm UL}\right)_{\rm BS}}(s)\Big|_{s=0}\!\right)^2.
\end{equation*}
The mean of the random variable $\left(I^{\rm UL}\right)_{\rm UE}$ conditioned on the random distance $\rho$ can be expressed by following a similar methodology. Hence, the expressions in Lemma \ref{MeanInterferenceUL} are derived.
\vspace{-10pt}
\section{Proof of Proposition \ref{Proposition1}}\label{Appendix1}
By leveraging the orthogonality property of the LMMSE estimator, we have $\mathbb{E}\left[\hat{g}_{0i}\left(e_i\right)^*\right]=0$, where $g_{0i}=\hat{g}_{0i}+e_i$. Hence, the estimation variance is given by $\sigma^2_{e_i}\triangleq\mathbb{E}\left[\left(e_i\right)^2\right]=1-\mathbb{E}\left[\left|\hat{g}_{0i} \right|^2 \right]$where, by using \eqref{VarianceProof}
\begin{equation*}
	\mathbb{E}\left[\left|\hat{g}_{0i} \right|^2 \right]\!=\! \left(\!\frac{\sqrt{\Lambda P\ell(r_{i}(\rho))}}{\Lambda P\ell(r_{i}(\rho))\!+\!N_0\!+\!\mathbb{E}(\mathcal{I}_i|\rho)}\!\right)^2\!\!\mathbb{E}\!\left[\!\left(\widetilde{y}_i\right)^*\!\widetilde{y}_i\right].
\end{equation*}
The final expression follows by evaluating the expectation.
\vspace{-10pt}

\section{Proof of Lemma \ref{VarianceLI}}\label{ProofAppendix}
The baseband equivalent received pilot signal for the LI CE at the $i$-th port of the typical UE, is given by
\begin{equation*}
	\mathbf{y}_{\rm LI}^u = \sqrt{\Lambda_u P_u(r_{i}(\rho))v_{\rm LI}}h^u_{\rm LI} \mathbf{X}_0+\sum\nolimits_{\substack{t\in\mathbb{N}^+ \\ y_t\in\Psi\backslash x_0}}\sqrt{P_u(\|\overline{y_t}\|)\ell\left(r_{i}(\|y_t\|)\right)}g_{ti}\mathbf{X}_t+\eta_0,
\end{equation*}
where $\mathbf{X}_0$ is a deterministic $\Lambda_u\times 1$ training symbol vector \cite{KAY}, $\mathbf{X}_t\stackrel{d}{\sim}\mathcal{CN}\left(\mathbf{0}_{\Lambda_u\times 1},\mathbf{I}_{\Lambda_u}\right)$ is a pilot symbol vector from the interfering UE at $x_t\in\Psi\backslash x_0$, with channel $g_{ti}\stackrel{d}{\sim}\mathcal{CN}(0,1)$, and $\eta_0\stackrel{d}{\sim}\mathcal{CN}\left(\mathbf{0}_{\Lambda_u\times 1},N_0\mathbf{I}_{\Lambda_u}\right)$ is the AWGN vector. By adopting the low-complexity LMMSE estimator, the estimate of $h^u_{\rm LI}$ is given by
\begin{equation*}
	\hat{h}_{\rm LI}|_{\rho}\!=\!\frac{\sqrt{\Lambda_u P_u(r_{i}(\rho))v_{\rm LI}}\sigma^2}{\Lambda_u P_u(r_{i}(\rho))\ell\left(r_{i}(\rho)\right)\sigma^2\!+\!N_0\!+\!\mathbb{E}\left(\left(\mathcal{I}_i^{\rm DL}\right)_{\rm UE}|\rho\right)}\widetilde{y}_{\rm LI}^u,
\end{equation*}
where $\mathbb{E}\left((\mathcal{I}_i^{\rm DL})_{\rm UE}|\rho\right)$ is the mean of interference that is caused by active interfering UEs and observed at the typical UE, that are given by \eqref{MeanDLUE} and \eqref{MeanULBS}. Thus, the variance of the CE error for the LI link at the $i$-th port of the typical UE and the tagged BS, can be derived by following a similar methodology as in Appendix \ref{Appendix1}.
\vspace{-10pt}
\section{Proof of Lemma \ref{Lemma1}}\label{Appendix2}
According to the adopted LMMSE-based CE technique, the channel between the $i$-th port of the typical UE and the tagged BS, can be indicated as $\widehat{g}_{{\rm 0}i}=g_{{\rm 0}i}+e_i$, where $g_{{\rm 0}i}\stackrel{d}{\sim}\mathcal{CN}(0,\sigma^2)$ and $e_i\stackrel{d}{\sim}\mathcal{CN}\left(0,\sigma^2_{e_i}\right)$, and thus, $\widehat{g}_{{\rm 0}i}\stackrel{d}{\sim}\mathcal{CN}\left(0,\widetilde{\sigma}_i^2\right)$ with $\widetilde{\sigma}_i^2=\sigma^2(1-\mu_i^2)+\sigma^2_{e_i}$. Then, the amplitude of the estimated channels, $|\widehat{g}_{{\rm 0}i}|$, is Rayleigh distributed, with pdf 
\begin{equation}\label{pdf}
	f_{|\widehat{g}_{{\rm 0}i}|}(\tau) = \frac{2\tau}{\widetilde{\sigma}_i^2}\exp\left(-\frac{\tau^2}{\widetilde{\sigma}_i^2}\right),
\end{equation}
with $\mathbb{E}\left[|\widehat{g}_{{\rm 0}i}|^2\right]=\widetilde{\sigma}_i$. Owing to the capability of the FA's ports to be arbitrarily close to each other, the channels $\{\widehat{g}_{{\rm 0}1},\dots,\widehat{g}_{{\rm 0}N}\}$. More specifically, the amplitude of the estimated channel $|\widehat{g}_{{\rm 0}2}|$, conditioned on $\rho$ and $|\widehat{g}_{{\rm 0}1}|$, follows a Rice distribution i.e.,
\begin{equation*}
	f_{\left|\widehat{g}_{{\rm 0}2}\right|\big|\left|\widehat{g}_{{\rm 0}1}\right|}(\tau_2|\rho,x_0,y_0)\!=\!\frac{2\tau_2}{\widetilde{\sigma}_i^2}\exp\left(-\frac{\tau_2^2+\mu^2_2(x_0^2+y_0^2)}{\widetilde{\sigma}_i^2}\right) I_0\left(\frac{2\mu_2\sqrt{x_0^2+y_0^2}\tau_2}{\widetilde{\sigma}_i^2}\right),
\end{equation*}
where $\tau_2\geq 0$. Then, by substituting $\tau_1\rightarrow\sqrt{x_0^2+y_0^2}$ and based on the fact that $x_0,y_0,\left|\widehat{g}_{{\rm 0}2}\right|,\dots,\left|\widehat{g}_{{\rm 0}N}\right|$ are all independent between each other, the joint pdf of the estimated channels, conditioned on $\left|\widehat{g}_{{\rm 0}1}\right|$, can be expressed as
\begin{equation*}
	f_{\left|\widehat{g}_{{\rm 0}2}\right|,\dots,\left|\widehat{g}_{{\rm 0}N}\right|\big|\left|\widehat{g}_{{\rm 0}1}\right|}(\tau_2,\dots,\tau_{N}|\rho,\tau_1)=\prod_{i\in\mathcal{N}}\frac{2\tau_i}{\widetilde{\sigma}_i^2}\exp\left(-\frac{\tau_i^2+\mu^i_2\tau_1^2}{\widetilde{\sigma}_i^2}\right)I_0\left(\frac{2\mu_i\tau_1\tau_i}{\widetilde{\sigma}_i^2}\right).
\end{equation*}
Hence, the final expression can be obtained by un-conditioning the above expression with the pdf of $\left|\widehat{g}_{{\rm 0}1}\right|$ given in \eqref{pdf} i.e., $f_{\left|\widehat{g}_{{\rm 0}2}\right|,\dots,\left|\widehat{g}_{{\rm 0}N}\right|\big|\left|\widehat{g}_{{\rm 0}1}\right|}(\tau_2,\dots,\tau_{N}|\rho,\tau_1)f_{\left|\widehat{g}_{{\rm 0}1}\right|}(\tau_1)$, which gives the desired expression.

By leveraging the derived expression for the joint pdf which is given by \eqref{PDF}, the joint cdf of $\{\widehat{g}_{{\rm 0}1},\dots,\widehat{g}_{{\rm 0}N}\}$ follows directly from the definition i.e.,
\begin{align}
	&F_{\left|\widehat{g}_{{\rm 0}1}\right|,\dots,\left|\widehat{g}_{{\rm 0}N}\right|}(\tau_1,\dots,\tau_{N}|\rho)\triangleq\int_{0}^{\tau_1}\cdots\int_{0}^{\tau_{N}}f_{\left|\widehat{g}_{{\rm 0}1}\right|,\dots,\left|\widehat{g}_{{\rm 0}N}\right|}(\tau_1,\dots,\tau_{N}|\rho){\rm d}\tau_1\cdots{\rm d}\tau_{N}\nonumber\\&\!=\!\int_0^{\tau_1}\!\frac{2t_1}{\widetilde{\sigma}_1^2}\exp\!\left(-\frac{t_1^2}{\widetilde{\sigma}_1^2}\right)\!\prod_{j\in\mathcal{N}}\Bigg[\int_0^{\tau_j}\frac{2t_j}{\widetilde{\sigma}_j^2}\exp\left(-\frac{t_j^2+\mu^2_jt_1^2}{\widetilde{\sigma}_j^2}\right)I_0\left(\frac{2\mu_jt_1t_j}{\widetilde{\sigma}_j^2}\right){\rm d}t_j\Bigg]{\rm d}t_1,\label{Proof1}
\end{align}
where \eqref{Proof1} is derived with the use of \cite[2.20]{Marcumbook}, since the integral inside the product operator is an integration over the pdf of a Ricean random variable. The final expression can be derived by using the transformation $t=\frac{t_1^2}{\widetilde{\sigma}_1^2}$, which concludes the proof.

\vspace{-10pt}
\section{Proof of Theorem \ref{Theorem1}}\label{Appendix3}
Based on the adopted port selection scheme, described by \eqref{Association}, the outage performance for the DL transmission can be evaluated as follows
\begin{align}
	\mathcal{P}_o^{\rm DL}(\theta) &= \mathbb{P}\!\left[\!\max_{i\in\mathcal{N}}\!\Bigg\lbrace\!\frac{\frac{P}{r^a_{i}(\rho)}\left| \widehat{g}_{{\rm 0}i}\right|^2}{\mathcal{I}^{\rm DL}_i\!+\!\Sigma^{\rm DL}_i}\!\Bigg\rbrace\!<\!\theta\right]\nonumber\\&=\mathbb{P}\!\Bigg[\frac{\frac{P}{r^a_{1}(\rho)}\left| \widehat{g}_{{\rm 0}1}\right|^2}{\mathcal{I}^{\rm DL}_1\!+\!\Sigma^{\rm DL}_1}<\theta,\dots,\frac{\frac{P}{r^a_{N}(\rho)}\left| \widehat{g}_{{\rm 0}N}\right|^2}{\mathcal{I}^{\rm DL}_N\!+\!\Sigma^{\rm DL}_N}<\!\theta\Bigg]\nonumber\\
	&=\mathbb{E}\left[\mathbb{P}\!\Bigg[\!\!\left| \widehat{g}_{{\rm 0}\mathtt{i}}\right|^2\!<\!\Theta_{\rm DL}(\mathtt{i})\Big|\rho\Bigg]\right],
\end{align}
where $\Theta_{\rm DL}(i) = \frac{\theta r_i^{a}(\rho)}{P}\left(\mathcal{I}^{\rm DL}_i\!+\!\Sigma^{\rm DL}_i\right)$ and $\Sigma^{\rm DL}_i=\sigma_{e^u_{\rm LI}}^2\mathcal{I}_{\rm LI}^{\rm DL}\!+\!\frac{P}{r^a_{{i}}(\rho)}\sigma^2_{e_{i}}\!+\!N_0$. The conditional outage performance i.e., $\mathcal{P}_o^{\rm DL}(\theta|\rho,\mathcal{I}^{\rm DL}_{\mathtt{i}})=\mathbb{P}\!\Big[\!\!\left| \widehat{g}_{{\rm 0}\mathtt{i}}\right|^2\!<\!\Theta_{\rm DL}(\mathtt{i})\big|\rho\Big]$, of the considered system model, is given by substituting $\tau_1=\sqrt{\Theta_{\rm DL}(1)}$,$\dots$, $\tau_N=\sqrt{\Theta_{\rm DL}(N)}$ into the joint cdf that is given in Lemma \ref{Lemma1}, and by un-conditioning with the interference distribution i.e., $f_\mathcal{I}(x,r)$. Finally, by un-conditioning the derived expression with the pdf of the distance from the typical UE to its serving BS, the final expression can be derived. Following a similar procedure as described above, the UL outage probability can be derived. Due to space limitation, the analysis for the UL transmission is omitted.

\end{document}